\documentclass[letterpaper, 11 pt]{article}  
\usepackage{amsthm}

\usepackage[left=2.2cm, right=2.2cm, top=3cm, bottom=3cm]{geometry}             

\usepackage[pdftex]{graphicx}
\usepackage{amsmath} 
\usepackage{amssymb}  
\usepackage{color}
\usepackage{enumerate}
\usepackage{gensymb}
\usepackage{multirow}
\usepackage{mathtools}
\usepackage{booktabs}
\usepackage{epstopdf}
\usepackage{algorithm,algcompatible}
\usepackage[normalem]{ulem}

\usepackage{url}
\usepackage{bm}

\usepackage{booktabs}

\newtheorem{theorem}{Theorem}
\newtheorem{remark}{Remark}
\newtheorem{proposition}{Proposition}

\makeatletter
\newsavebox\myboxA
\newsavebox\myboxB
\newlength\mylenA

\newcommand*\xoverline[2][0.75]{%
    \sbox{\myboxA}{$\m@th#2$}%
    \setbox\myboxB\null
    \ht\myboxB=\ht\myboxA%
    \dp\myboxB=\dp\myboxA%
    \wd\myboxB=#1\wd\myboxA
    \sbox\myboxB{$\m@th\overline{\copy\myboxB}$}
    \setlength\mylenA{\the\wd\myboxA}
    \addtolength\mylenA{-\the\wd\myboxB}%
    \ifdim\wd\myboxB<\wd\myboxA%
       \rlap{\hskip 0.5\mylenA\usebox\myboxB}{\usebox\myboxA}%
    \else
        \hskip -0.5\mylenA\rlap{\usebox\myboxA}{\hskip 0.5\mylenA\usebox\myboxB}%
    \fi}
\makeatother

\newcommand{\R}{\mathbb{R}}

\usepackage{tabularx}
\usepackage{booktabs}

\newtheorem{definition}{Definition}
\newtheorem{assumption}{Assumption}

\makeatletter
\let\NAT@parse\undefined
\makeatother
\usepackage{hyperref}
\usepackage[doi=false,isbn=false,defernumbers=true,style=ieee,backend=bibtex,maxbibnames=1000]{biblatex}

\bibliography{AMPC.bib}

\usepackage{accents}
\newcommand{\ubar}[1]{\underaccent{\bar}{#1}}

\usepackage{balance}

\title{\LARGE \bf Adaptive MPC under Time Varying Uncertainty: Robust and Stochastic}

\author{Monimoy~Bujarbaruah,
        Xiaojing Zhang,
        Marko Tanaskovic,
        and~Francesco Borrelli\thanks{M. Bujarbaruah, X. Zhang and F. Borrelli are with the Department of Mechanical Engineering, University of California Berkeley, Berkeley, CA 94720 USA; e-mail: \{monimoyb, xiaojing.zhang, fborrelli\}@berkeley.edu. M. Tanaskovic is with Univerzitet Singidunum, Belgrade, Serbia; e-mail: mtanaskovic@singudunum.ac.rs.}
}

\begin{document}

\maketitle
  \thispagestyle{empty}
\pagestyle{empty}

\maketitle

\begin{abstract}

This paper deals with the problem of formulating an adaptive Model Predictive Control strategy for constrained uncertain systems. We consider a linear system, in presence of bounded time varying additive uncertainty. The uncertainty is decoupled as the sum of a process noise with known bounds, and a time varying offset that we wish to identify. The time varying offset uncertainty is assumed unknown point-wise in time. Its domain, called the Feasible Parameter Set, and its maximum rate of change are known to the control designer. As new data becomes available, we refine the Feasible Parameter Set with a Set Membership Method based approach, using the known bounds on process noise. We consider two separate cases of robust and probabilistic constraints on system states, with hard constraints on actuator inputs. In both cases, we \emph{robustly} satisfy the imposed constraints for all possible values of the offset uncertainty in the Feasible Parameter Set. By imposing adequate terminal conditions, we prove recursive feasibility and stability of the proposed algorithms. The efficacy of the proposed robust and stochastic Adaptive MPC algorithms is illustrated with detailed numerical examples. 

\end{abstract}

\section{Introduction}
Model Predictive Control (MPC) is an established control methodology  for dealing with constrained, and possibly uncertain systems 
\cite{mayne2000constrained, morari1999model, borrelli2017predictive}. Primary challenges in MPC design include presence of disturbances and/or unknown model parameters. Disturbances can be handled by means of robust or chance constraints, and such methods are generally well understood 
\cite{kothare1996robust, schwarm1999chance, carson2013robust, langson2004robust, Goulart2006, limon2010robust,  tempo2012randomized, ZhangKamgGeorghiouGoulLyg_Aut16}. In this paper, we are looking into methods for addressing the  challenge posed by model uncertainties when adaptation is introduced in the design. 

If the actual model of a system is unknown, adaptive control strategies have been applied for meeting control objectives and ensuring a system's stability. Adaptive control for unconstrained systems has been widely studied and is generally well-understood \cite{sastry2011adaptive, krstic1995nonlinear}. In recent times, this concept of online model adaptation has been extended to MPC controller design for systems subject to both robust and probabilistic constraints \cite{fukushima2007adaptive, dehaan2007adaptive, adetola2009adaptive, aswani2013provably, chowdhary2013concurrent,  wang2014adaptive, marafioti2014persistently,tanaskovic2014adaptive,weiss2014robust, ostafew2014learning,hewing2017cautious,lorenzen2017adaptive, limon2017learning, heirung2017dual, 2017arXiv171207548T, bujarbaruahAdapFIR, hernandez2018stabilizing, soloperto2018learning}. 

The vast majority of literature on adaptive MPC for uncertain systems has focused on robust constraint satisfaction.
For linear systems, works such as \cite{aswani2013provably, hernandez2018stabilizing,soloperto2018learning} have typically focused on improving performance with the adapted models (e.g. low closed-loop cost), while the constraints are satisfied robustly for all possible modeling errors and all disturbances realizations. Here, the domain (support) of the model uncertainty is  not adapted in real-time, which may lead to suboptimal controllers.
The work of \cite{tanaskovic2013adaptive,tanaskovic2014adaptive,2017arXiv171207548T} deal with both time invariant and time varying system uncertainty in Finite Impulse Response (FIR) domain and prove recursive feasibility and stability \cite[Chapter~12]{borrelli2017predictive} of the proposed approaches. However, such FIR parametrization restricts application to primarily slow and stable systems. In \cite{lorenzen2017adaptive, lorenzenAutomaticaAMPC, luAccCannon}, Linear Parameter Varying (LPV) models are considered, and recursive feasibility of robust constraints and closed-loop stability properties are ensured in presence of unknown, but \emph{time-invariant} parameters only. 
The authors in \cite{fukushima2007adaptive} also  formulate an adaptive MPC strategy for an LPV system using the concept of \emph{comparison models}, but do not consider any disturbances or process noise.
For nonlinear systems, works such as \cite{dehaan2007adaptive, adetola2009adaptive, wang2014adaptive} propose robust adaptive MPC algorithms, but since these require construction of invariant sets \cite[Chapter 10]{borrelli2017predictive} for such systems, they are computationally demanding. 

Literature on adaptive MPC for systems with \emph{probabilistic} constraints is   more limited. The work in \cite{hewing2017cautious, ostafew2014learning} use data driven approaches for real-time model learning together with a stochastic MPC controller, but without guarantees on feasibility and stability. In \cite{heirung2017dual, bujarbaruahAdapFIR} recursively feasible adaptive stochastic MPC algorithms are presented, but for static input-output system models only. To the best of our knowledge, no adaptive MPC framework has been presented in literature that ensures recursive feasibility and stability for systems in state-space under probabilistic constraints. 
 
In this paper, we  build on the work of \cite{tanaskovic2014adaptive,2017arXiv171207548T, bujarbaruahAdapFIR},
and propose a unified and tractable \emph{Adaptive MPC} framework for systems represented by state-space models, that can take into account both robust and probabilistic state constraints, and hard input constraints, while guaranteeing recursive feasibility and stability. 
Specifically, we consider linear systems that are subject to bounded additive uncertainty, which is composed of: $(i)$ a process noise, and $(ii)$ an unknown, but bounded offset that we try to estimate.
Given an initial estimate of the offset's domain, we iteratively refine it using a Set Membership Method based approach \cite{tanaskovic2014adaptive}, as new data becomes available. In order to design an MPC controller with the unknown offset, we make sure the constraints on states and inputs are satisfied for all feasible offsets at a time instant. Here a ``feasible offset" is an offset belonging to the current estimation of the offset's domain. 
As the feasible offset domain is updated with data progressively, we obtain an on-line adaptation in the MPC algorithm. Furthermore, the offset uncertainty present in the system is considered time varying and its maximum rate of change is assumed bounded and known \cite{wang2014adaptive,2017arXiv171207548T}. The main contributions of this paper can be summarized as follows:
\begin{itemize}
    \item We propose a Set Membership Method based model adaptation algorithm to estimate and update the time varying offset uncertainty, using a so-called Feasible Parameter Set. The model adaptation algorithm guarantees containment of the true offset uncertainty in the Feasible Parameter Set at all times. 
    This extends the works of \cite{tanaskovic2013adaptive, tanaskovic2014adaptive, 2017arXiv171207548T,lorenzen2017adaptive,2018arXiv180409831B} to time varying model uncertainties in state space.

    \item We propose an adaptive MPC framework for systems perturbed by such additive time varying offset uncertainty and process noise. The framework handles robust and chance constraints on system states respectively, with hard input constraints, while using data to progressively obtain offset uncertainty adaptation. With appropriately chosen terminal conditions, we guarantee recursive feasibility and Input to State Stability (ISS) of the proposed adaptive MPC algorithms, which is an addition to the work of \cite{hewing2017cautious, ostafew2014learning, soloperto2018learning, koller2018learning}. Compared to \cite{dehaan2007adaptive, adetola2009adaptive, wang2014adaptive}, computation of terminal invariant sets is simpler, as we focus on linear systems. Moreover, as opposed to \cite{aswani2013provably,hernandez2018stabilizing, soloperto2018learning}, we utilize the model adaptation information in real-time for modifying constraints.
\end{itemize}
The paper is organized as follows: in Section~\ref{sec:prob_for} we formulate the optimization problems to be solved and also define the imposed constraints. The offset uncertainty adaptation framework is presented in Section~\ref{sec:offset_adap}. We propose the Adaptive Robust MPC algorithm in Section~\ref{sec:robust_MPC} and Adaptive Stochastic MPC algorithm in Section~\ref{sec:stoch_MPC}. The feasibility and stability properties of the aforementioned algorithms are discussed in Section~\ref{sec:feasandstab}. We then present detailed numerical simulations in Section~\ref{sec:simul}. 
\section{Problem Formulation} \label{sec:prob_for}
Given an initial state $x_S$, we consider uncertain linear time-invariant systems of the form:
\begin{equation}\label{eq:unc_system}
	x_{t+1} = Ax_t + Bu_t + E\theta^a_t+ w_t,~x_0 = x_S,
\end{equation}
where $x_t\in \mathbb{R}^{n}$ is the state at time step $t$, $u_t\in\mathbb{R}^{m}$ is the input, and $A$ and $B$ are known system matrices of appropriate dimensions. At each time step $t$, the system is affected by an i.i.d.\ random process noise $w_t \in \mathbb W \subset \mathbb{R}^{n}$, whose probability distribution function (PDF) is assumed known, or can be estimated empirically from data \cite{waterman1978estimation, masson2006inferring}.  For simplicity , $\mathbb W$ is assumed to be a hyperrectangle containing zero as: 
\begin{align}\label{eq:w_bounds}
    \mathbb{W} = \{w: -\bar{w} \leq w \leq \bar{w} \}.
\end{align}
We also consider the presence of $\theta^a_t \in \mathbb{R}^p$, a bounded, time varying offset uncertainty, which enters the system through the constant known matrix $E \in \mathbb{R}^{n \times p}$. 

\begin{remark}
In reality, the additive uncertainty in the system could be difficult to  split into two parts as considered in \eqref{eq:unc_system}. However, such a decomposition enables us to deal with parametric model uncertainties. Although, we have formulated the problem with only additive uncertainty in \eqref{eq:unc_system}, where $A$ and $B$ are known matrices, one can also upper bound effect of parametric uncertainties in $A$ and $B$ with an additive uncertain term (similar to offset $E\theta^a_t$) and propagate the system dynamics \eqref{eq:unc_system} with a chosen set of nominal $\bar{A},\bar{B}$ matrices. 
\end{remark}

\begin{assumption} \label{ass:rate_theta}
We assume the true offset $\theta^a_t$ to be time varying. The bounds on the rate of change of this offset are known and given by $\theta^a_t - \theta^a_{t-1} = \Delta \theta^a_t \in \mathcal{P}$,~for all $t \geq 0$, where the set
\begin{align}\label{eq:d_bounds}
    \mathcal{P} = \{\Delta \theta^a \in \mathbb{R}^p : K^\theta \Delta \theta^a \leq l^\theta, K^\theta \in \mathbb{R}^{r_\theta \times p}, l^\theta \in \mathbb{R}^{r_\theta}\}.
\end{align}
\end{assumption}

\begin{assumption} \label{ass:omega_b}
We also assume that the true offset $\theta^a_t$ lies within a known and nonempty polytope $\Omega$ at all times, which contains zero in its interior. That is,
\begin{align}\label{eq:Omega}
   \theta^a_t \in \Omega,~\forall t \geq 0,~\text{where,~} \Omega = \{\theta:\mathcal{H}^{\theta}_{0}\theta \leq h^{\theta}_{0}\}, 
\end{align}
for matrices  $\mathcal{H}^\theta_0 \in \mathbb{R}^{r_0 \times p}$ and ${h}^\theta_0 \in \mathbb{R}^{r_0}$. 
\end{assumption}

\subsection{Constraints}
In this paper we study two different cases of constraint satisfaction, namely $(i)$ robust constraint on states and hard constraints on inputs, and $(ii)$ chance constraints on states and hard constraints on inputs. We define $C \in \mathbb{R}^{s \times n}$, $G \in \mathbb{R}^{s \times n}$, $D \in \mathbb{R}^{s \times m}$, $b \in \mathbb{R}^s$, $h \in \mathbb{R}^s$, $H_u \in \mathbb{R}^{o \times m}$ and $h_u \in \mathbb{R}^o$. We can then write the constraints $\forall t \geq 0$ as:
\begin{subequations} \label{eq:constraints_nominal}
\begin{align}
    \mathbb{Z}_1 & := \{(x,u): C x + D u \leq b\},\label{eq: con_robust}\\
	\mathbb{Z}_2 & := \{(x,u): \mathbb{P}(Gx 
	\leq h)  \geq 1-\alpha,~ {H}_u u \leq h_u \} \label{eq:chance_con_com},
\end{align}
\end{subequations}
where $\alpha \in (0,1)$ is the admissible probability of constraint violation. We assume the above state and input constraint sets are compact and they contain the origin. This assumption is key for the stability proof in Section~\ref{sec:feasandstab}.

\subsection{Infinite Horizon Optimization Problems}
Our goal is to design controllers that solve two infinite horizon optimal control problems, one with constraints $\mathbb{Z}_1$ and the other one with $\mathbb{Z}_2$. They are defined as follows:
\begin{equation}\label{eq:OP_inf_R}
	\begin{array}{clll}
		\hspace{0cm} 
	\displaystyle\min_{u_0,u_1(\cdot),\ldots} & \displaystyle\sum\limits_{t\geq0} \ell \left(\bar{x}_t, \bar{u}_t \right) \\[1ex]
		\text{s.t.}  & x_{t+1} = Ax_t + Bu_t + E\theta^a_t +w_t,\\[1ex]
		& \bar{x}_{t+1} = A\bar{x}_t + B \bar{u}_t + E \theta^a_t,\\[1ex]
		& C x_t + D u_t \leq b,\ \forall w_t \in \mathbb W,\ \\[1ex]
		&  x_0 =  x_S, \ \bar{x}_0 =  x_S,\\ & t=0,1,\ldots,
	\end{array}\tag{P1}
\end{equation}
and
\begin{equation}\label{eq:OP_inf_S}
	\begin{array}{clll}
		\hspace{0cm} 
	\displaystyle\min_{u_0,u_1(\cdot),\ldots} & \displaystyle\sum\limits_{t\geq0} \ell \left(\bar{x}_t, \bar{u}_t \right) \\[1ex]
		\text{s.t.}  & x_{t+1} = Ax_t + Bu_t + E\theta^a_t +w_t,\\[1ex]
		& \bar{x}_{t+1} = A\bar{x}_t + B\bar{u}_t + E \theta^a_t,\\[1ex]
		& \mathbb{P}(Gx_t 
	\leq h)  \geq 1-\alpha, \\[1ex]
	& H_u u_t \leq h_u,\ \\[1ex]
	&  x_0 =  x_S, \ \bar{x}_0 =  x_S,\\ & t=0,1,\ldots,
	\end{array}\tag{P2}
\end{equation}
where $\theta^a_t$ is the time varying offset present in the system, $\bar{x}_t$ is the disturbance-free nominal state and $\bar{u}_t = {u}_t(\bar{x}_t)\in \mathbb{R}^m $ denotes the corresponding nominal input. The nominal state is utilized to obtain the nominal cost, which is minimized in optimization problems \eqref{eq:OP_inf_R} and \eqref{eq:OP_inf_S}. We point out that, as system \eqref{eq:unc_system} is uncertain, the optimal control problems \eqref{eq:OP_inf_R} and \eqref{eq:OP_inf_S} consist of finding input policies $[u_0,u_1(\cdot),u_2(\cdot),\ldots]$, where $u_t: \mathbb{R}^{n}\ni x_t \mapsto u_t = u_t(x_t)\in \mathbb{R}^m$ are feedback policies. 
We approximate solutions to optimization problems \eqref{eq:OP_inf_R} and \eqref{eq:OP_inf_S} by solving corresponding finite time constrained optimal control problems in a receding horizon fashion. 

In this paper, we assume that the offset $\theta^a_t$ in \eqref{eq:OP_inf_R} and \eqref{eq:OP_inf_S} is \emph{not known exactly}. Therefore, we propose a parameter estimation framework to refine our knowledge of $\theta^a_t$ as more data is collected, thus introducing \emph{adaptation}. 

\section{Uncertainty Adaptation}\label{sec:offset_adap}
The domain of feasible offset $\theta^a_t$ is denoted by $\Theta_t$ at time step $t$, and is called the \emph{Feasible Parameter Set} \cite{tanaskovic2014adaptive}. The goal is to ensure that constraints \eqref{eq: con_robust} and \eqref{eq:chance_con_com} are satisfied for all $\theta_t \in {
\Theta}_t$. This guarantees constraint satisfaction in presence of the true unknown offset $\theta^a_t \in \Theta_t$. 
Our initial estimate for $\Theta_0$ is $\Omega$ from Assumption~\ref{ass:omega_b}, i.e., $\Theta_0 = \Omega$. The Feasible Parameter Set is then adapted at every time-step as new measurements are available, utilizing Assumption~\ref{ass:rate_theta} and Assumption~\ref{ass:omega_b}. Based on only the measurements at time step $t$, we denote the potential domain of the feasible offset at time step $t$, $\mathcal{S}^t_t$ as: 
\begin{align*}
    \mathcal{S}^t_t= ~\{\theta_t \in \mathbb{R}^p: -\bar{w}+\ubar \nu \leq -x_t + Ax_{t-1} + Bu_{t-1} + E\theta_t  \leq \bar{w} + \bar{\nu}\},
\end{align*}
where bounds $\bar{w}$ are given by \eqref{eq:w_bounds}, and from \eqref{eq:d_bounds}, we apply: 
\begin{equation}\label{eq:vbar_updates}
\begin{aligned}
    & \ubar \nu = \min_\nu \{ E\nu : K^\theta \nu\leq l^\theta\},\\
    & \bar \nu = \max_\nu \{ E\nu :  K^\theta \nu\leq l^\theta\}.
\end{aligned}
\end{equation}
Now, for any $q\leq t$, the feasible set of offsets for time step $t$, based on information until time step $q$, is obtained as: 
\begin{align*}
    \mathcal{S}^q_t= ~\{\theta_t \in \mathbb{R}^p: -\bar{w}+(t-q+1)\ubar \nu \leq -x_q + Ax_{q-1} + Bu_{q-1} + E\theta_t \leq \bar{w} + (t-q+1)\bar{\nu}\},
 \end{align*}
Using all the above information until time step $t$, we obtain the Feasible Parameter Set at time step $t$, as:
\begin{align*}
\Theta_t = \Omega \cap \big(\bigcap \limits_ {q=1,2,\dots,t} \mathcal{S}^q_t\big).
\end{align*}
The above Feasible Parameter Set at time step $t$ can be written as:
\begin{align}\label{eq:fps_matrix}
    \Theta_t = \{\theta_t \in \mathbb{R}^p: \mathcal{H}^{\theta}_t \theta_t \leq h^{\theta}_t\},
\end{align}
where $\mathcal{H}^\theta_t \in \mathbb{R}^{r_t \times p}$ and $h^\theta_t \in \mathbb{R}^{r_t}$, $r_t = r_0 + 2t$ is the number of facets in the Feasible Parameter Set polytope $\Theta_t$ at any given $t$. As new data is obtained at the next time step $(t+1)$, it can be proven that \cite{2017arXiv171207548T}:
\begin{equation}\label{eq:fps_cl_update}
\begin{aligned}
    \mathcal{H}^{\theta}_{t+1} & = [ (\mathcal{H}^{\theta}_t)^\top, -E^\top, +E^\top]^\top \in \mathbb{R}^{r_{t+1} \times p},\\ 
    h^{\theta}_{t+1} & = \begin{bmatrix} h^{\theta}_t + \Delta h^{\theta}_t \\ -x_{t+1}+Ax_t+Bu_t+\bar{w}-\ubar{\nu} \\ x_{t+1}-Ax_t-Bu_t+\bar{w}+\bar{\nu} \end{bmatrix} \in \mathbb{R}^{r_{t+1}},\\
    \Delta h^{\theta}_t & = \begin{bmatrix} \mathbf{0}_{r_0}^\top,-\ubar \nu^\top, \bar{\nu}^\top, \dots, -\ubar \nu^\top, \bar{\nu}^\top \end{bmatrix}^\top \in \mathbb{R}^{r_t}.
\end{aligned}
\end{equation}

\begin{proposition}
Assume that \eqref{eq:w_bounds} and Assumption~\ref{ass:rate_theta} hold. Then the Feasible Parameter Set  obtained using \eqref{eq:fps_matrix}--\eqref{eq:fps_cl_update} is nonempty and contains the true offset at all times, i.e, $\Theta_t\neq\emptyset$ and $\theta^a_t \in \Theta_t$ for all $t \geq 0$.
\end{proposition}
\begin{proof}
See Appendix.
\end{proof}

\section{Adaptive Robust MPC} \label{sec:robust_MPC}
In this section we present formulation of the Adaptive Robust MPC algorithm. We approximate the solutions to the infinite horizon optimal control problem \eqref{eq:OP_inf_R} by solving a finite horizon problem in a receding horizon fashion.

\subsection{Robust  MPC Problem}
The MPC controller has to solve the following finite horizon robust optimal control problem at each time step:
\begin{equation}\label{MPC_R_intrac} 
	\begin{aligned}
		\min_{U_t(\cdot)} & \sum_{k=t}^{t+N-1} \ell(\bar{x}_{k|t}, \bar{u}_{k|t}) + Q(\bar{x}_{t+N|t}) \\
		~~\text{s.t.} &~~  x_{k+1|t} = Ax_{k|t} + Bu_{k|t} + E\theta_{k|t} + w_{k|t},\\
		& ~~\bar{x}_{k+1|t} = A\bar{x}_{k|t} + B \bar{u}_{k|t} + E \bar{\theta}_t,\\
		&~~C x_{k|t} + D u_{k|t} \leq b,\\
		&~~x_{t+N|t} \in \mathcal{X}^R_N,\\
		&~~ \forall \theta_{k|t} \in \Theta_{k|t},~\forall w_{k|t} \in \mathbb{W},\\
		&~~\forall k = \{t,\ldots,t+N-1\},\\
	    &~~ x_{t|t} = x_t, \ \bar{x}_{t|t} = x_t, \bar{\theta}_t \in \Omega,
	\end{aligned}
\end{equation}
where $x_t$ is the measured state at time step $t$, $x_{k|t}$ is the prediction of state at time step $k$, obtained by applying predicted input policies $[u_{t|t},\dots,u_{k-1|t}]$ to system~\eqref{eq:unc_system}, and $\{\bar{x}_{k|t}, \bar{u}_{k|t}\}$ with $\bar{u}_{k|t} = u_{k|t}(\bar{x}_{k|t})$ denote the disturbance-free nominal state and corresponding input respectively. We use a nominal point estimate of offset, $\bar{\theta}_t \in \Omega$ to propagate the nominal trajectory. The predicted Feasible Parameter Sets $\Theta_{k|t}$ are elaborated in the following section. Notice, the above minimizes the nominal cost, comprising of positive definite stage cost $\ell(\cdot, \cdot)$ and terminal cost $Q(\cdot)$ functions. The terminal constraint $\mathcal{X}^R_N$ and terminal cost $Q(\cdot)$ are introduced to ensure feasibility and stability properties of the MPC controller \cite{mayne2000constrained, borrelli2017predictive}, as we highlight in Section~\ref{sec:feasandstab}.   

\begin{remark}\label{rem:nom_off}
One may design point estimates $\bar{\theta}_t$ of the offset for performance improvement, i.e, lower cost in \eqref{MPC_R_intrac}. Following \cite{lorenzenAutomaticaAMPC}, one option is to construct the nominal offset estimate $\bar{\theta}_t$ recursively with Least Mean Square filter as
\begin{subequations}\label{eq:theta_loren}
    \begin{align}
        & \tilde{\theta}_t = \bar{\theta}_{t-1} + \mu E^\top (x_t - \bar{x}_{t|t-1}), \\
        & \bar{\theta}_t = \mathrm{Proj}_{\Omega} (\tilde{\theta}_t), \label{eq:proj_theta}
    \end{align}
\end{subequations}
where $\mathrm{Proj}(\cdot)$ denotes the Euclidean projection operator, and scalar $\mu \in \R$ can be chosen such that $\frac{1}{\mu} > \Vert E \Vert ^2$. 
\begin{proposition}\label{prop:offset_error}
If $\sup_{t \geq 0} \Vert x_t \Vert < \infty$ and $\sup_{t \geq 0} \Vert u_t \Vert < \infty$, then $\bar{\theta}_t \in \Omega$ and 
\begin{align}\label{eq:pred_err}
    \sup_{\tilde{m} \in \mathbb{N}, w_t \in \mathbb{W}, \bar{\theta}_0 \in \Omega} \frac{\sum \limits_{t=0}^{\tilde{m}} \Vert \tilde{x}_{t+1|t} \Vert^2 }{\frac{1}{\mu} \Vert \bar{\theta}_0 - \theta_0^a \Vert^2 + \sum \limits_{t=0}^{\tilde{m}} \Vert w_t \Vert^2 } \leq 1,
\end{align}
where $\tilde{x}_{t+1|t} = Ax_t + Bu_t - \bar{x}_{t+1|t}$ is the one step prediction error, ignoring the effect of $w_t$ in closed-loop.
\end{proposition}
\begin{proof}
See Appendix. 
\end{proof}
With bound \eqref{eq:pred_err} on prediction error, finite gain $\ell_2$ stability of the resulting MPC algorithm can be trivially proven by following \cite[Theorem~14]{lorenzenAutomaticaAMPC}, \cite[Theorem~3.2]{kouvaritakis2016model}. However, since we only focus on the robust constraint satisfaction aspect of \eqref{MPC_R_intrac}, we will use the nominal offset estimate $\bar{\theta}_t = \mathbf{0}_{p \times 1}$ for all $t \geq 0$ in the subsequent sections. 
\end{remark}

\subsection{Predicted Feasible Parameter Sets}\label{sec:OL_fps}
These \emph{Predicted Feasible Parameter Sets} are constructed along an MPC horizon at time step $t$, when the measurement at next time step $(t+1)$ is yet to be available. 
\begin{definition}(Predicted Feasible Parameter Sets)
The \emph{Predicted Feasible Parameter Sets} at any time step $t$, are the predicted feasible domains of the true offset $\theta^a$ over a prediction horizon of length $N$, based on the information until time step $t$. These sets are denoted as $\Theta_{k|t}=\{\theta  \in \mathbb{R}^{p}: \mathcal{H}^\theta_{k|t} \theta \leq h^\theta_{k|t}\}$ for all $k \in \{t,t+1,\dots,t+N-2\}$, where
\begin{subequations}\label{eq:ol_fps}
\begin{align}
    \mathcal{H}^{\theta}_{k+1|t} & = \mathcal{H}^{\theta}_{k|t} \in \mathbb{R}^{r_t \times p},\\
    h^{\theta}_{k+1|t} & = h^{\theta}_{k|t} + \begin{bmatrix} \mathbf{0}_{r_0}^\top,-\ubar \nu^\top, \bar{\nu}^\top,.., -\ubar \nu^\top, \bar{\nu}^\top \end{bmatrix}^\top,
\end{align}
\end{subequations}
with the terminal condition, 
\begin{align}\label{eq:fps_termc}
    \Theta_{t+N|t} = \Omega,
\end{align}
where $\Omega$ is defined in Assumption~\ref{ass:omega_b}.
\end{definition}
In principle, the Predicted Feasible Parameter Sets in \eqref{eq:ol_fps} are formed after measuring $x_t$ at any time step $t$, and expanding the obtained (from \eqref{eq:fps_matrix}) Feasible Parameter Set $\Theta_t$ over the entire horizon of length $N$, incorporating parameter rate bounds \eqref{eq:d_bounds}. 

\begin{proposition} 
The Predicted Feasible Parameter Sets satisfy the property $\Theta_{k|t+1} \subseteq \Theta_{k|t}$, for all $k \in \{t+1,t+2,\dots,t+N\}$.
\end{proposition}

\begin{proof}
See Appendix.
\end{proof}

\subsection{Control Policy}\label{sec:control_param}

Note that optimizing over policies $U(\cdot)$ in \eqref{MPC_R_intrac} is an intractable problem, as it involves searching over an infinite dimensional function space. Therefore, we restrict ourselves to the affine disturbance feedback parametrization \cite{Goulart2006,lofberg2003minimax} for control synthesis. For all $k \in \{t,\dots,t+N-1\}$ over the MPC horizon (of length $N$), the control policy is given as:
\begin{equation}\label{eq:inputParam_DF_OL}
	u_{k|t}(x_{k|t}): u_{k|t} = \sum \limits_{j=t}^{k-1}M_{k,j|t} (w_{j|t} + E\theta_{j|t})  + v_{k|t},
\end{equation}
where $M_{k|t}$ are the \emph{planned} feedback gains at time step $t$ and $v_{k|t}$ are the auxiliary inputs. Let us define $\mathbf{w}_t = [w^\top_{t|t},\dots,w^\top_{t+N-1|t}]^\top \in \mathbb{R}^{nN}$, $\pmb{\theta}_t = [\theta^\top_{t|t},\dots,\theta^\top_{t+N-1|t}]^\top \in \mathbb{R}^{pN}$ and \textbf{E} = $\textnormal{diag}(E,E,
\dots,E) \in \mathbb{R}^{nN \times pN}$. Then the sequence of predicted inputs from \eqref{eq:inputParam_DF_OL} can be stacked together and compactly written as $\mathbf{u}_t = \mathbf{M}_t(\mathbf{w}_t + \mathbf{E}\pmb{\theta}_t) + \mathbf{v}_t$ at any time step $t$, where $\mathbf{M}_t\in \mathbb{R}^{mN \times nN}$ and $\mathbf{v}_t \in \mathbb{R}^{mN}$ are:
\begin{equation*}
\begin{aligned}
  & \mathbf{M}_t  = \begin{bmatrix}0& \cdots&\cdots&0\\
  M_{t+1,t}& 0 & \cdots & 0\\
  \vdots &\ddots& \ddots &\vdots\\
  M_{t+N-1,t}& \cdots& M_{t+N-1,t+N-2}& 0
  \end{bmatrix}, \\ 
  & \mathbf{v}_t = [v_{t|t}^\top, \cdots, \cdots, v_{t+N-1|t}^\top]^\top.
\end{aligned}
\end{equation*} 

\subsection{Tractable Reformulation}
Using Section~\ref{sec:OL_fps} and Section~\ref{sec:control_param}, we solve the following tractable reformulation of robust MPC problem \eqref{MPC_R_intrac}: 
\begin{equation} \label{eq:MPC_R_fin}
	\begin{aligned}
	  J^\star_R&(t,x_t)  :=	\\
	& \min_{\mathbf{M}_t, \mathbf{v}_t} ~~ \sum_{k=t}^{t+N-1} \ell(\bar{x}_{k|t}, v_{k|t}) + Q(\bar{x}_{t+N|t}) \\
		& ~~~\text{s.t.}~~~~~    x_{k+1|t} = Ax_{k|t} + Bu_{k|t} + E\theta_{k|t} + w_{k|t},\\
		& ~~~~~~~~~~~~\bar{x}_{k+1|t} = A\bar{x}_{k|t} + Bv_{k|t},\\
		&~~~~~~~~~~~~u_{k|t} = \sum \limits_{j=t}^{k-1}M_{k,j|t} (w_{j|t} + E\theta_{j|t})  + v_{k|t},\\
		&~~~~~~~~~~~~C x_{k|t} + D u_{k|t} \leq b,\\
	    &~~~~~~~~~~~~x_{t+N|t} \in \mathcal{X}^R_N,\\
	    & ~~~~~~~~~~~~ \forall \theta_{k|t} \in \Theta_{k|t},~\forall w_{k|t} \in \mathbb{W},\\
        &~~~~~~~~~~~~ \forall k = \{t,\ldots,t+N-1\},\\
        	&~~~~~~~~~~~~x_{t|t} = x_t, \ \bar{x}_{t|t} = x_t.
	\end{aligned}
\end{equation}
We use state feedback to construct terminal set $\mathcal{X}^R_N = \{x \in \mathbb{R}^n: Y_R x \leq z_R,~Y_R \in \mathbb{R}^{r_R \times n},~z_R \in \mathbb{R}^{r_R}\}$, which is the maximal robust positive invariant set \cite{kolmanovsky1998theory} obtained with a state feedback controller $u=Kx$, dynamics \eqref{eq:unc_system} and constraints \eqref{eq: con_robust}. This set has the properties:
\begin{equation}\label{eq:term_set_DF}
    \begin{aligned}
    &\mathcal{X}^R_N \subseteq \{x|(x,Kx) \in \mathbb{Z}_1\},\\
    &(A+BK)x + w + E\theta \in \mathcal{X}^R_N,~\\
    &\forall x\in \mathcal{X}^R_N,~\forall w \in \mathbb{W},~\forall \theta \in \Omega.
    \end{aligned}
\end{equation}
Fixed point iteration algorithms to numerically compute \eqref{eq:term_set_DF} can be found in \cite{borrelli2017predictive, kouvaritakis2016model}. Notice that \eqref{eq:MPC_R_fin} is a \emph{time varying} convex optimization problem with $\infty-$number of constraints. An efficient way to reformulate \eqref{eq:MPC_R_fin} is shown in the Appendix. After solving \eqref{eq:MPC_R_fin}, in closed-loop, we apply,
\begin{equation}\label{eq:inputCL_DF}
u_t(x_t) = 	u^\star_{t|t} = v^\star_{t|t}
\end{equation}
to system \eqref{eq:unc_system}. We then resolve the problem again at the next $(t+1)$-th time step, yielding a receding horizon strategy. 


\begin{algorithm}[h!]
    \caption{
Adaptive Robust MPC
    }
    \label{alg1}
        \begin{algorithmic}[1]
      \STATE Set $t=0$; initialize Feasible Parameter Set $\Theta_0 = \Omega$; 
    
      \STATE Compute the parameter rate of change bounds $\ubar{\nu}$ and $\bar{\nu}$ from \eqref{eq:vbar_updates}; 
      
      \STATE Form Predicted Feasible Parameter Sets $\Theta_{k|t}$ for $k=\{t,\dots,t+N\}$ using \eqref{eq:ol_fps} and \eqref{eq:fps_termc};      
      
      \STATE Compute $v^\star_{t|t}$ from \eqref{eq:MPC_R_fin} and apply $u_t = v^\star_{t|t}$ to \eqref{eq:unc_system};
      
      \STATE Obtain $x_{t+1}$, and update $\Theta_{t+1}$ as given in \eqref{eq:fps_cl_update};
            
       \STATE Set $t = t+1$, and return to step~3.
        \end{algorithmic}
\end{algorithm}


\section{Adaptive Stochastic MPC}\label{sec:stoch_MPC}
In this section we present the formulation of the Adaptive Stochastic MPC algorithm. Similar as before, we approximate \eqref{eq:OP_inf_S} by solving a finite horizon problem in receding horizon fashion. For parametrization of control policies, we use the same affine disturbance policy as in Section~\ref{sec:control_param}. 

\subsection{MPC Problem}
We use Bonferroni's inequality \cite{farina2016stochastic} to approximate the joint chance constraints on states \eqref{eq:chance_con_com}, given as:
\begin{equation}\label{eq:chance_con}
    \mathbb{P}(g^\top_jx_t 
	\leq h_j)  \geq 1-\alpha_j, H_u u_t \leq h_u,~\forall j \in \{1,\dots,s\}, 
\end{equation}
where $\alpha_j \in (0,1)$, $\sum \limits_{j=1}^{s} \alpha_j = \alpha$ and $g_j^\top$ denotes the $j$-th row of matrix $G$ for all $j \in \{1,\dots, s\}$. To ensure satisfaction of state constraints in \eqref{eq:chance_con}, it is sufficient to ensure \cite{korda2011strongly}:
\begin{align}\label{eq:chance_con_conditional}
\mathbb{P}(g^\top_jx_{t+1} \leq h_j|x_t)  \geq 1-\alpha_j,~j \in \{1,\dots, s\},~\forall t \geq 0.
\end{align} 
Therefore, the stochastic MPC controller has to solve the following optimal control problem at each time step:
\begin{equation} \label{eq:idealMPC_sto}
	\begin{aligned}
    & \min_{\mathbf{M}_t, \mathbf{v}_t}  \sum_{k=t}^{t+N-1} \ell(\bar{x}_{k|t}, v_{k|t}) + Q(\bar{x}_{t+N|t}) \\
    & ~~~\text{s.t.}~~  x_{k+1|t} = Ax_{k|t} + Bu_{k|t} + E\theta_{k|t} + w_{k|t},\\
		& ~~~~~~~~~ \bar{x}_{k+1|t} = A\bar{x}_{k|t} + Bv_{k|t},\\
		& ~~~~~~~~~ \mathbb{P}(g^\top_jx_{k+1|t} 
	\leq h_j|x_{k|t})  \geq 1-\alpha_j,\\
	& ~~~~~~~~~ u_{k|t} = \sum \limits_{j=t}^{k-1}M_{k,j|t} (w_{j|t} + E\theta_{j|t})  + v_{k|t},\\
	& ~~~~~~~~~ H_u u_{k|t} \leq h_u,\\
	& ~~~~~~~~~x_{t+N|t} \in \mathcal{X}^S_N,\\
	& ~~~~~~~~~ \forall [\theta_{t|t},\ldots,\theta_{k|t}]\in\Theta_{t|t}\times\ldots\times\Theta_{k|t},\\
	& ~~~~~~~~~ \forall [w_{t|t},\ldots,w_{k-1|t}]\in\mathbb{W}^{k-t},\\
	& ~~~~~~~~~ \forall k = \{t,\ldots,t+N-1\},~ \forall j \in \{1,\dots,s\},\\
	& ~~~~~~~~~x_{t|t} = x_t, \ \bar{x}_{t|t} = x_t,
	\end{aligned}
\end{equation}
where the terminal constraint $\mathcal{X}^S_N$ and the terminal cost function $Q(\cdot)$ are introduced to ensure feasibility and stability properties of the MPC controller \cite{mayne2000constrained, borrelli2017predictive}.

\subsection{Chance Constraint Reformulation}
The chance constraints in \eqref{eq:idealMPC_sto} can be reformulated as
\begin{align}\label{eq:chance_con_mpcreform}
    & g_j^\top (Ax_{k|t} + Bu_{k|t} + E\theta_{k|t}) \leq h_j-F^{-1}_{g^\top_jw}(1-\alpha_j), \nonumber \\ 
    & \forall [\theta_{t|t},\ldots,\theta_{k|t}]\in\Theta_{t|t}\times\ldots\times\Theta_{k|t}, \nonumber \\  
    & \forall [w_{t|t},\ldots,w_{k-1|t}]\in\mathbb{W}^{k-t}, \nonumber \\
    &\forall j\in \{1,\dots,s\},~\forall k \in \{t,\dots,t+N-1\},
\end{align}
where $F^{-1}_{g^\top_jw}(\cdot)$ is the left quantile function of $g^\top_j w$. From \eqref{eq:unc_system} and \eqref{eq:inputParam_DF_OL}, using $\mathbf{w}_t = [w^\top_{t|t},\dots,w^\top_{t+N-1|t}]^\top \in \mathbb{R}^{nN}$, $\pmb{\theta}_t = [\theta^\top_{t|t},\dots,\theta^\top_{t+N-1|t}]^\top \in \mathbb{R}^{pN}$, \textbf{E} = $\textnormal{diag}(E,E,
\dots,E) \in \mathbb{R}^{nN \times pN}$ and $\mathbf{M}_t \in \mathbb{R}^{mN \times nN}$ from Section~\ref{sec:control_param}, we can write, $x_{k|t} = A^{k-t}x_{t|t} + \mathcal{B}_k(\mathbf{v}_t+\mathbf{M}_t\mathbf{w}_t + \mathbf{M}_t\mathbf{E} \pmb{\theta}_t ) + \mathcal{C}_k(\mathbf{w}_t + \mathbf{E}\pmb{\theta}_t)$, where $\mathcal{B}_k = [A^{k-t-1}B,\dots,B,\mathbf{0}_{n \times m},\dots,\mathbf{0}_{n \times m}] \in \mathbb{R}^{n \times mN}$ and $\mathcal{C}_k = [A^{k-t-1},\dots,I_{n \times n},\mathbf{0}_{n \times n},\dots,\mathbf{0}_{n \times n}] \in \mathbb{R}^{n \times nN}$ \cite{korda2011strongly}. Using this to rewrite LHS of \eqref{eq:chance_con_mpcreform}, we obtain:
\begin{align}\label{eq:chance_con_lhs_combined}
    g_j^\top (Ax_{k|t} + Bu_{k|t} + E\theta_{k|t}) = g_j^\top(A^{k+1-t}x_t + \mathcal{B}_{k+1}\mathbf{v}_t) + g_j^\top(\mathcal{B}_{k+1}\mathbf{M}_t+A \mathcal{C}_k)(\mathbf{w}_t + \mathbf{E}\pmb{ \theta}_t) + g^\top_j(E\theta_{k|t}).
\end{align}
Denoting, $\Gamma^\star_{k|t}(\mathbf{M}_t) =  \max_{\mathbf{w}_t,\pmb{\theta}_t} g_j^\top(\mathcal{B}_{k+1}\mathbf{M}_t+A \mathcal{C}_k)(\mathbf{w}_t + \mathbf{E}\pmb{ \theta}_t) + g^\top_j(E\theta_{k|t})$, and with \eqref{eq:chance_con_lhs_combined}, we rewrite \eqref{eq:chance_con_mpcreform} as
\begin{equation*}
    \begin{aligned}
& g_j^\top(A^{k+1-t}x_t + \mathcal{B}_{k+1}\mathbf{v}_t) \leq h_j-\Gamma^\star_{k|t}(\mathbf{M}_t) - F^{-1}_{g^\top_jw}(1-\alpha_j),\\ 
&\forall \theta_{k|t} \in \Theta_{k|t},~\forall j \in \{1,\dots,s\},~\forall k \in \{t,\dots,t+N-1\}. 
    \end{aligned}
\end{equation*}
At any given $t$, the chance constraints \eqref{eq:chance_con_mpcreform} are thus robustly satisfied for all offsets $[\theta_{t|t},\ldots,\theta_{k|t}]\in\Theta_{t|t}\times\ldots\times\Theta_{k|t}$ and $[w_{t|t},\ldots,w_{k-1|t}]\in\mathbb{W}^{k-t}$ for all $k \in \{t,\dots,t+N-1\}$.

\subsection{Tractable MPC Problem}
Using the previous results,  \eqref{eq:idealMPC_sto} is equivalent to the following deterministic optimization problem: 
\begin{equation} \label{eq:idealMPC_sto_trac}
	\begin{aligned}
	     J^\star_S&(t,x_t)  :=	\\	
		& \min_{\mathbf{M}_t, \mathbf{v}_t}  \sum_{k=t}^{t+N-1} \ell(\bar{x}_{k|t}, v_{k|t}) + Q(\bar{x}_{t+N|t}) \\
		& ~~~\text{s.t.}~~   x_{k+1|t} = Ax_{k|t} + Bu_{k|t} + E\theta_{k|t} + w_{k|t},\\
		& ~~~~~~~~~  \bar{x}_{k+1|t} = A\bar{x}_{k|t} + Bv_{k|t},\\
	    & ~~~~~~~~~ g_j^\top(A^{k+1-t}x_{t|t} + \mathcal{B}_{k+1}\mathbf{v}_t) \leq h_j-\Gamma^\star_{k|t}(\mathbf{M}_t) -F^{-1}_{g^\top_jw}(1-\alpha_j),\\[-2ex]
		& ~~~~~~~~~ u_{k|t} = \sum \limits_{j=t}^{k-1}M_{k,j|t} (w_{j|t} + E\theta_{j|t})  + v_{k|t},\\
		& ~~~~~~~~~  H_u u_{k|t} \leq h_u,\\
	    & ~~~~~~~~~  x_{t+N|t} \in \mathcal{X}^S_N,\\
	    &~~~~~~~~~  \forall \theta_{k|t} \in \Theta_{k|t},~\forall k = \{t,\ldots,t+N-1\},\\
	    & ~~~~~~~~~ \forall j \in \{1,\dots,s\},\\
	    & ~~~~~~~~~ x_{t|t} = x_t, \ \bar{x}_{t|t} = x_t, 
	\end{aligned}
\end{equation}
where the terminal set $\mathcal{X}^S_N = \{x \in \mathbb{R}^n: Y_S x \leq z_S,~Y_S \in \mathbb{R}^{r_S \times n},~z_S \in \mathbb{R}^{r_S}\}$ has the properties:
\begin{equation}\label{eq:term_set_stoch}
    \begin{aligned}
    & (A+BK)x+E\theta+w \in \mathcal{X}^S_N,\\
    & H_u (Kx) \leq h_u,\\
    & g^\top_j(A+BK)x \leq -g^\top_j(E\theta)+ h_j -F^{-1}_{g^\top_jw}(1-\alpha_j),\\
    & \forall j \in \{1,\dots,s\},~\forall \theta \in \Omega,~\forall w \in \mathbb{W},~\forall x \in \mathcal{X}^S_N.
    \end{aligned}
\end{equation}
Notice that \eqref{eq:idealMPC_sto_trac} is a \emph{time varying} convex optimization problem. An efficient way to reformulate \eqref{eq:idealMPC_sto_trac} is shown in the Appendix. After solving \eqref{eq:idealMPC_sto_trac}, in closed-loop, we apply the first input,
\begin{equation}\label{eq:inputCL_DF_sto}
	u_t(x_t) = u^\star_{t|t} = v^\star_{t|t}
\end{equation}
to system \eqref{eq:unc_system}. We then resolve the problem again at the next $(t+1)$-th time step, yielding a receding horizon control design. 

\begin{algorithm}[h!]
    \caption{
Adaptive Stochastic MPC
    }
    \label{alg2}
        \begin{algorithmic}[1]
      \STATE Set $t=0$; initialize Feasible Parameter Set $\Theta_0 = \Omega$; 
    
      \STATE Compute the parameter rate of change bounds $\ubar{\nu}$ and $\bar{\nu}$ from \eqref{eq:vbar_updates}; 
      
      \STATE Form Predicted Feasible Parameter Sets $\Theta_{k|t}$ for $k=\{t,\dots,t+N\}$ using \eqref{eq:ol_fps} and \eqref{eq:fps_termc};      
      
      \STATE Compute $v^\star_{t|t}$ from \eqref{eq:idealMPC_sto_trac} and apply $u_t = v^\star_{t|t}$ to \eqref{eq:unc_system};
      
      \STATE Obtain $x_{t+1}$, and update $\Theta_{t+1}$ as given in \eqref{eq:fps_cl_update};
            
       \STATE Set $t = t+1$, and return to step~3.
        \end{algorithmic}
\end{algorithm}
\section{Feasibility and Stability Guarantees}\label{sec:feasandstab}
In this section we discuss feasibility and stability properties of Algorithm~\ref{alg1} and Algorithm~\ref{alg2}.

\begin{assumption}\label{stagecost} 
The stage cost $\ell(\cdot, \cdot)$ in \eqref{eq:MPC_R_fin} and \eqref{eq:idealMPC_sto_trac} is chosen as $\ell(\bar{x}_{k|t}, v_{k|t}) = \bar{x}_{k|t}^\top P \bar{x}_{k|t} + v_{k|t}^\top R v_{k|t}$ for some $P=P^\top \succ 0$ and $R=R^\top \succ 0$, which is continuous and positive definite. 
\end{assumption}
\begin{assumption}\label{ass: termcost}
The terminal cost $Q(\cdot)$ in \eqref{eq:MPC_R_fin} (or \eqref{eq:idealMPC_sto_trac}) is chosen as a Lyapunov function in the terminal set $\mathcal{X}^R_N$ (or $\mathcal{X}^S_N$) for the nominal closed-loop system ${x}^+ = (A+BK){x}$, for all $\bar{x} \in \mathcal{X}_N^R$ (or $\mathcal{X}_N^S$). That is, $Q((A+BK)x) - Q(x) \leq -x^\top(P+K^\top R K)x$.
\end{assumption}

\subsection{Feasibility}
\begin{theorem}\label{thm1}
Let Assumptions~\ref{ass:rate_theta}-{\ref{ass:omega_b}} hold and consider the robust  optimization problem \eqref{eq:MPC_R_fin}. Let this optimization problem be feasible at time step $t=0$ with $\Theta_0 = \Omega$ with $\Omega$ defined in   Assumption~\ref{ass:omega_b}. 
Assume the Feasible Parameter Set $\Theta_{t}$ in \eqref{eq:MPC_R_fin} is adapted based on \eqref{eq:fps_matrix}-\eqref{eq:fps_cl_update}.
Then, \eqref{eq:MPC_R_fin} remains feasible at all time steps $t\geq 0$, if the state $x_t$ is  obtained by applying  the closed-loop MPC control law \eqref{eq:MPC_R_fin}-\eqref{eq:inputCL_DF} to system \eqref{eq:unc_system}.
\end{theorem}

\begin{proof}
Let the optimization problem \eqref{eq:MPC_R_fin} be feasible at time step $t$. Let us denote the corresponding optimal input policies as $[u^\star_{t|t},u^\star_{t+1|t}(\cdot),\dots,u^\star_{t+N-1|t}(\cdot)]$. Assume the MPC controller $u^\star_{t|t}$ is applied to \eqref{eq:unc_system} in closed-loop and $\Theta_{t+1}$ is updated according to \eqref{eq:fps_cl_update}. Consider a candidate policy sequence at the next time step as:
\begin{align}\label{eq:feas_seq_next_DF_sto}
    U_{t+1}(\cdot) = [u^\star_{t+1|t}(\cdot),\dots,u^\star_{t+N-1|t}(\cdot),Kx_{t+N|t+1}].
\end{align}
We observe the following two facts: $(i)$ from Proposition~2, $\Theta_{k|t+1} \subseteq \Theta_{k|t}$, for all $k \in \{t+1,t+2,\dots,t+N\}$, and $(ii)$ from \eqref{eq:term_set_DF}, terminal set $\mathcal{X}_N^R$ is robust positive invariant for all $w \in \mathbb{W}$, and for all $\theta \in \Omega$, with state feedback controller $Kx$. Using $(i)$ we conclude $[u^\star_{t+1|t}(\cdot),u^\star_{t+2|t}(\cdot),\dots,u^\star_{t+N-1|t}(\cdot)]$ is an $(N-1)$ step feasible sequence at $(t+1)$ (excluding terminal condition), since at previous time step $t$, it robustly satisfied all stage constraints in \eqref{eq:MPC_R_fin} for $\Theta_{k|t}$, for all $k \in \{t+1,t+2,\dots,t+N-1\}$. With this feasible sequence, $x_{t+N|t+1} \in \mathcal{X}^R_N$. Using $(ii)$ we conclude, appending the $(N-1)$ step feasible sequence with $Kx_{t+N|t+1}$ ensures $x_{t+N+1|t+1} \in \mathcal{X}^R_N$, satisfying the terminal constraint at $(t+1)$. 
\end{proof}

\subsection{Input to State Stability}
We denote the $N$-step robust controllable set \cite[Chapter 10]{borrelli2017predictive} to the terminal set $\mathcal{X}^R_N$ under the MPC policy \eqref{eq:inputCL_DF} by $\mathcal{X}_0^R$, which is compact and contains the origin. 
\begin{definition}(Input to State Stability \cite{li2018input}):
Consider system \eqref{eq:unc_system} in closed-loop with the MPC controller \eqref{eq:inputCL_DF}, obtained from \eqref{eq:fps_matrix}-\eqref{eq:fps_cl_update}-\eqref{eq:MPC_R_fin}, given by
\begin{align}\label{eq:cl_loop_system}
    x_{t+1} = Ax_t + Bv^\star_{t|t} + E\theta^a_t + w_t,~\forall t\geq 0.
\end{align}
The origin is defined as Input to State Stable (ISS), with a region of attraction $\mathcal{X}^R_0 \subset \mathbb{R}^{n}$, if there exists $\mathcal{K}_\infty$ functions $\alpha_1(\cdot)$, $\alpha_2(\cdot)$, $\alpha_3(\cdot)$, a $\mathcal{K}$ function $\sigma(\cdot)$ and a function $V(\cdot,\cdot): \mathbb{R} \times \mathcal{X}^R_0 \mapsto \mathbb{R}_{\geq 0}$ continuous at the origin such that, 
\begin{align*}
    & \alpha_1(\Vert x_t \Vert ) \leq V(t,x) \leq \alpha_2(\Vert x_t \Vert ),~\forall x \in \mathcal{X}^R_0,~\forall t \geq 0,\\
    & V(t+1, x_{t+1}) - V(t,x_t) \leq -\alpha_3(\Vert x_t \Vert) + \sigma(\Vert w_i + E\theta^a_i \Vert_{\mathcal{L}_\infty}),
\end{align*}
where $\Vert \cdot \Vert$ denotes the Euclidean norm and signal norm $\Vert d_i \Vert_{\mathcal{L}_\infty} = \sup_{i = \{0,1,\dots,t\}}\Vert d_i \Vert$. Function $V(\cdot,\cdot)$ is called an ISS Lyapunov function for \eqref{eq:cl_loop_system}.
\end{definition}

\vspace{1em}

\begin{theorem}\label{isstheorem}
Let Assumptions~\ref{ass:rate_theta}-\ref{ass: termcost} hold. Then, the optimal cost of \eqref{eq:MPC_R_fin}, i.e, $J^\star_R(\cdot,\cdot)$ is an ISS Lyapunov function for closed-loop system \eqref{eq:cl_loop_system}.
\end{theorem}

\begin{proof}
From Assumption~\ref{stagecost} we know that at time step $t$, $\alpha_1(\Vert x_t \Vert_2) \leq \ell(x_t,0) \leq J^\star_R(t,x_t)$ for some $\alpha_1(\cdot) \in \mathcal{K}_\infty$. Moreover, since \eqref{eq:MPC_R_fin} is a parametric QP, $J^\star_R(t,0) = 0$, and $\mathcal{X}_0^R$ is compact, using similar arguments as \cite[Theorem~23]{Goulart2006}, we know $J^\star_R(t,x_t) \leq \alpha_2(\Vert x_t \Vert_2)$ for some $\alpha_2(\cdot) \in \mathcal{K}_\infty$. Note that as opposed to \cite{Goulart2006}, our $\alpha_2(\cdot)$ is not obtained via Lipschitz continuity of the value function, since in our case, $V(t,x)$ is assumed continuous only at the origin. The existence of $\alpha_2(\cdot)$ is ensured by the compactness of the constraint sets in \eqref{eq:constraints_nominal}. Now say 
\begin{align}\label{eq:iss1}
J^\star_R(t,x_t) & = \sum \limits_{k=t}^{t+N-1} \ell(\bar{x}^\star_{k|t},v^\star_{k|t}) + Q(\bar{x}^\star_{t+N|t}), \nonumber \\
& = \ell(\bar{x}^\star_{t|t},v^\star_{t|t}) + q(\bar{x}^\star_{t+1|t}),
\end{align}
where $[\bar{x}^\star_{t|t},\dots,\bar{x}^\star_{t+N|t}]$ is the optimal predicted nominal trajectory under the optimal nominal input sequence $U_t^\star(\bar{x}_t) = [u^\star_{t|t}(\bar{x}_{t|t}),\dots,u^\star_{t+N-1|t}(\bar{x}_{t+N-1|t})]$ applied to nominal dynamics in \eqref{eq:MPC_R_fin}, and $q(\bar{x}^\star_{t+1|t})$ provides the total cost from $(t+1)$ to $(t+N)$ under policy $U_t^\star(\bar{x}_t)$. We proved that \eqref{eq:feas_seq_next_DF_sto} is a feasible policy sequence for \eqref{eq:MPC_R_fin} at time step $(t+1)$, where $x_{t+1} = \bar{x}_{t+1}$ is obtained with \eqref{eq:cl_loop_system}. With this feasible sequence, the optimal cost of \eqref{eq:MPC_R_fin} at $(t+1)$ is bounded as
\begin{align}\label{iss2}
    J^\star_R(t+1,x_{t+1}) & \leq \sum \limits_{k=t+1}^{t+N-1} \ell(\bar{x}_{k|t+1},v^\star_{k|t}) \nonumber + Q(\bar{x}_{t+N|t+1}), \nonumber \\
    & = q(\bar{x}_{t+1|t+1}),
\end{align}
where we have used Assumption~\ref{ass: termcost} and the feasible nominal trajectory  $\bar{x}_{k|t+1} = A^{k-t-1}(Ax_t + Bv^\star_{t|t} + w_t +  E\theta^a_t) + \sum \limits_{i=t+1}^{k-1} A^{k-1-i}B u^\star_{i|t}(\bar{x}_{k|t+1})$, for $k = \{t+2,\dots,t+N\}$. Moreover, we know that
\begin{align}\label{iss3}
    \bar{x}_{t+1|t+1}  = \bar{x}^\star_{t+1|t} + w_t + E\theta^a_t.
\end{align}
Combining \eqref{eq:iss1}--\eqref{iss3} we obtain,
\begin{align*}
    & J^\star_R(t+1,x_{t+1}) - J^\star_R (t,x_t) \nonumber \\
    & = q(\bar{x}^\star_{t+1|t} + w_t + E\theta^a_t) - \ell(\bar{x}^\star_{t|t},v^\star_{t|t}) - q(\bar{x}^\star_{t+1|t}),\\
    & \leq - \ell(\bar{x}^\star_{t|t},v^\star_{t|t}) + L_q \Vert w_t + E\theta^a_t \Vert,\\
    & \leq - \ell(\bar{x}^\star_{t|t},0) + L_q \Vert w_t + E\theta^a_t \Vert,\\
    & \leq -\alpha_3(\Vert x_t \Vert_2 ) + L_q\Vert w_i + E \theta_i^a \Vert_{\mathcal{L}_\infty},~\text{with } \alpha_1(\cdot) = \alpha_3(\cdot),
\end{align*}
where $q(\cdot)$ is $L_q$-Lipschitz, as $q(\cdot)$ is a sum of quadratic terms in compact \eqref{eq: con_robust}. Hence the origin of \eqref{eq:cl_loop_system} is ISS.
\end{proof}  
\begin{remark}
Feasibility and stability of Algorithm~\ref{alg2} can be proven in the exact same manner and hence is omitted.
\end{remark}
\section{Numerical Simulations}\label{sec:simul}

We consider the following infinite horizon optimal  control problem with unknown offset $\theta^a_t$ that satisfies Assumption~\ref{ass:rate_theta} and Assumption-\ref{ass:omega_b}:
\begin{equation}\label{eq:num_inf_DIS}
	\begin{array}{llll}
	&	\hspace{-0.cm}    	\displaystyle\min_{u_0,u_1(\cdot),\ldots}  \displaystyle\sum\limits_{t\geq0} \left \| \bar{x}_t \right\|_2^2 + 10 \left\| \bar{u}_t \right\|_2^2  \\[1ex]
		& ~~~~\text{s.t.}\\
		&~~~~~~~~~~~~x_{t+1} = Ax_t + Bu_t({x}_t) + E\theta^a_t +  w_t,\\[1ex]
		&~~~~~~~~~~~~\bar{x}_{t+1} = A\bar{x}_t + B\bar{u}_t + E\bar{\theta}_t,\\[1ex]
		&~~~~~~~~~~~~ \mathbb{P}\Big\{\begin{bmatrix}-5 \\ -2.5
		\end{bmatrix} \leq x_t \leq \begin{bmatrix}5 \\ 2.5 
		\end{bmatrix}\Big\} \geq 1-\alpha,\\[2ex]
		&~~~~~~~~~~~ -4 \leq u_t({x}_t) \leq 4,\\[0.5ex]
		& ~~~~~~~~~~~~\forall w_t \in \mathbb{W},~\forall {\theta}_t \in \Theta_t,\\[1.5ex]
		& ~~~~~~~~~~~~x_0 = x_S,\ \bar{x}_0 = x_S, \\ &~~~~~~~~~~~~ t=0,1,\ldots,
	\end{array}
\end{equation}
where $A = \begin{bmatrix} 1.2 & 1.5\\ 0 & 1.3\end{bmatrix}$ and $B = [0,1]^\top$, and Feasible Parameter Set $\Theta_t$ is updated based on \eqref{eq:fps_matrix}--\eqref{eq:fps_cl_update} for all time steps $t \geq 0$. For the robust MPC controller we pick $\alpha = 0$ and for the stochastic MPC $\alpha = 0.4$, which we split using Bonferroni's inequality as $\alpha_1 = \alpha_2 = \alpha/2$ for each of the two individual state constraints. Process noise $w_t \in \mathbb W = \{w \in \mathbb R^2: ||w||_{\infty} \leq 0.1 \}$. The initial Feasible Parameter Set is defined as $\Omega = \Theta_0 = \{ \theta \in \mathbb{R}^2: [-0.5, -0.5]^\top \leq \theta \leq [0.5, 0.5]^\top\}$. The true offset parameter $\theta^a_t$ is time varying, with rate bounded by the polytope $\mathcal{P} := [-0.05, -0.05]^\top \leq \Delta \theta^a_t \leq [0.05, 0.05]^\top$. For numerical simulations, we generate a true offset that starts from $\theta^a_0= [0.49, 0.49]^
\top$, and has a rate of change $\Delta \theta^a_t = [-0.0395, -0.0395]^\top$ as shown in Fig.~\ref{Fig:fps_adap_DIR}. The matrix $E \in \mathbb{R}^{2 \times 2}$ is picked as the identity matrix. The Adaptive Robust MPC in \eqref{eq:MPC_R_fin}, \eqref{eq:inputCL_DF}, and the Adaptive Stochastic MPC in \eqref{eq:idealMPC_sto_trac}, \eqref{eq:inputCL_DF_sto} are implemented with a control horizon of $N=6$, and the feedback gain $K$ in \eqref{eq:term_set_DF} and \eqref{eq:term_set_stoch}, is chosen to be the optimal LQR gain for system $x^+ = Ax + Bu$  with parameters $Q_\mathrm{LQR}=I$ and $R_\mathrm{LQR} = 10$. Initial state for both algorithms is $x_S = [-3.21, -0.25]^\top$.
 
Fig.~\ref{Fig:fps_adap_DIR} shows the recursive adaptation of the Feasible Parameter Set and time evolution of the true offset $\theta^a_t$. The true parameter lies within $\Omega$, and is always captured by $\Theta_t$ at all times. This evolution is kept identical for all simulation scenarios with both the algorithms. 
\begin{figure}[t]
    \centering
	\includegraphics[width= 12cm]{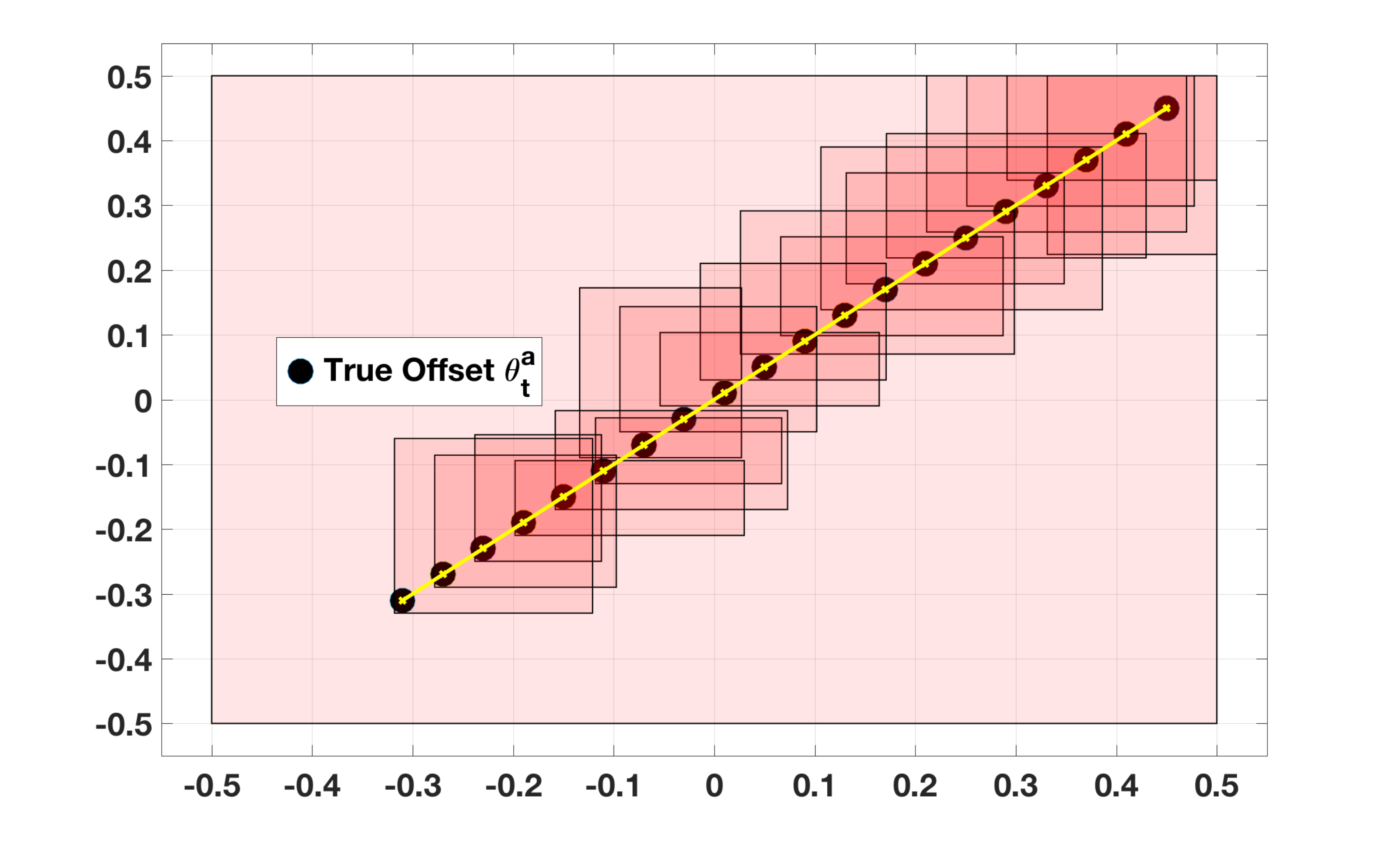}
    \caption{Feasible Parameter Set Evolution}
    \label{Fig:fps_adap_DIR}
\end{figure}
Fig.~\ref{Fig:MC_states_DIR} shows the Monte Carlo simulations for $100$ different sampled trajectories with Adaptive Robust MPC, which highlights satisfaction of constraints in \eqref{eq:num_inf_DIS} robustly $(\alpha = 0)$ for all feasible offset uncertainties $\theta_t \in \Theta_t$ and process noise $w_t \in \mathbb{W}$, for all $t \geq 0$. Such robust satisfaction of constraints is crucial for safety critical applications.
\begin{figure}[h!]
    \centering
	\includegraphics[width= 12cm]{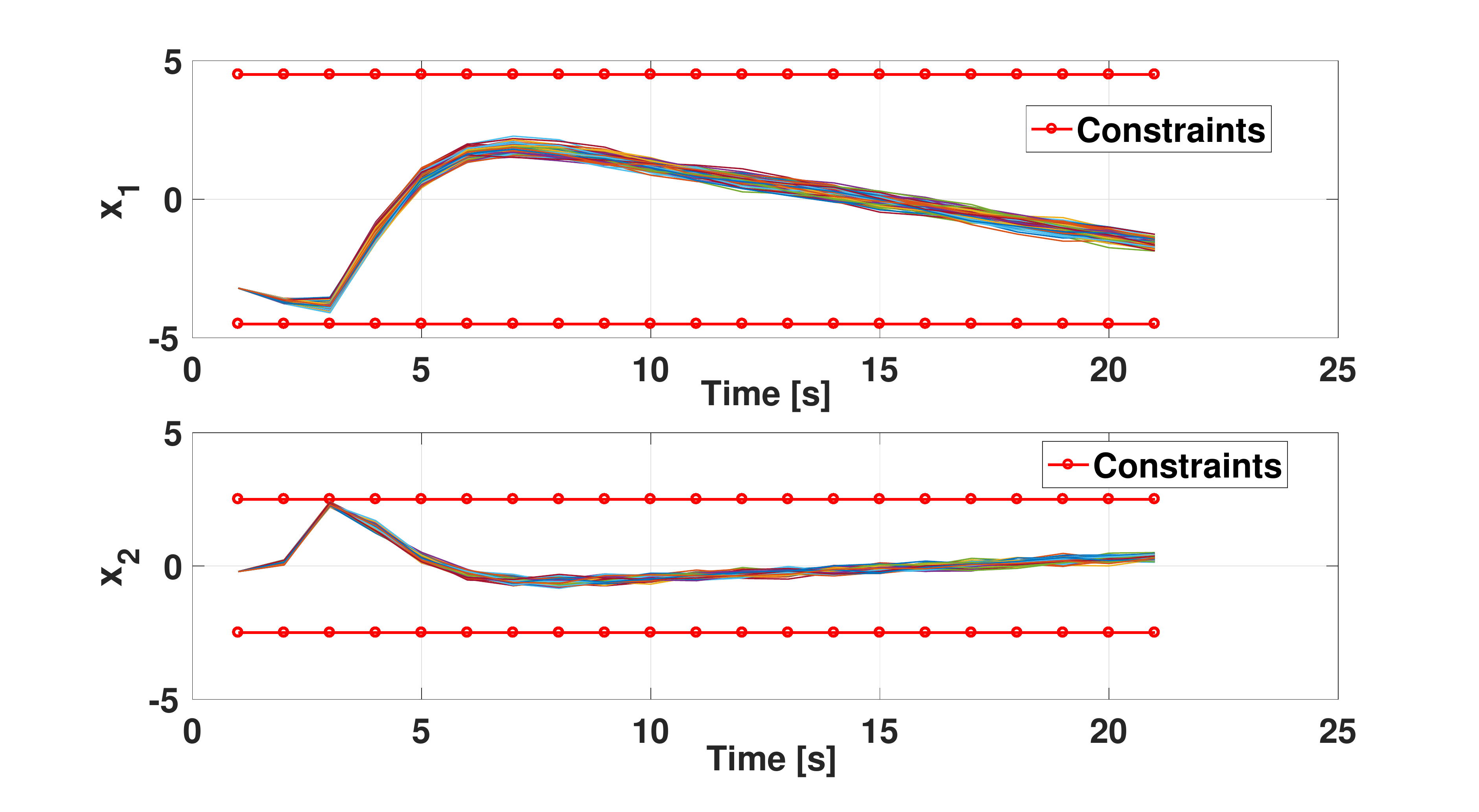}
    \caption{Monte Carlo Simulation of Robust Case}
    \label{Fig:MC_states_DIR}
\end{figure}

On the other hand, Adaptive Stochastic MPC could be applied in scenarios which are not safety critical, and where constraint violations are tolerated to improve performance (for example, lower closed-loop cost). Fig.~\ref{Fig:MC_states_DIS} shows Monte Carlo simulations for $100$ different sampled trajectories from Adaptive Stochastic MPC with the same process noise sequences as used for the previous example. This highlights satisfaction of chance constraints in \eqref{eq:num_inf_DIS} for all feasible offset uncertainties $\theta_t \in \Theta_t$ for all $t \geq 0$. The total empirical constraint violation probability is approximately $18\%$, which is lower than the allowed maximum value of $40\%$.
\begin{figure}[h!]
    \centering
	\includegraphics[width= 12cm]{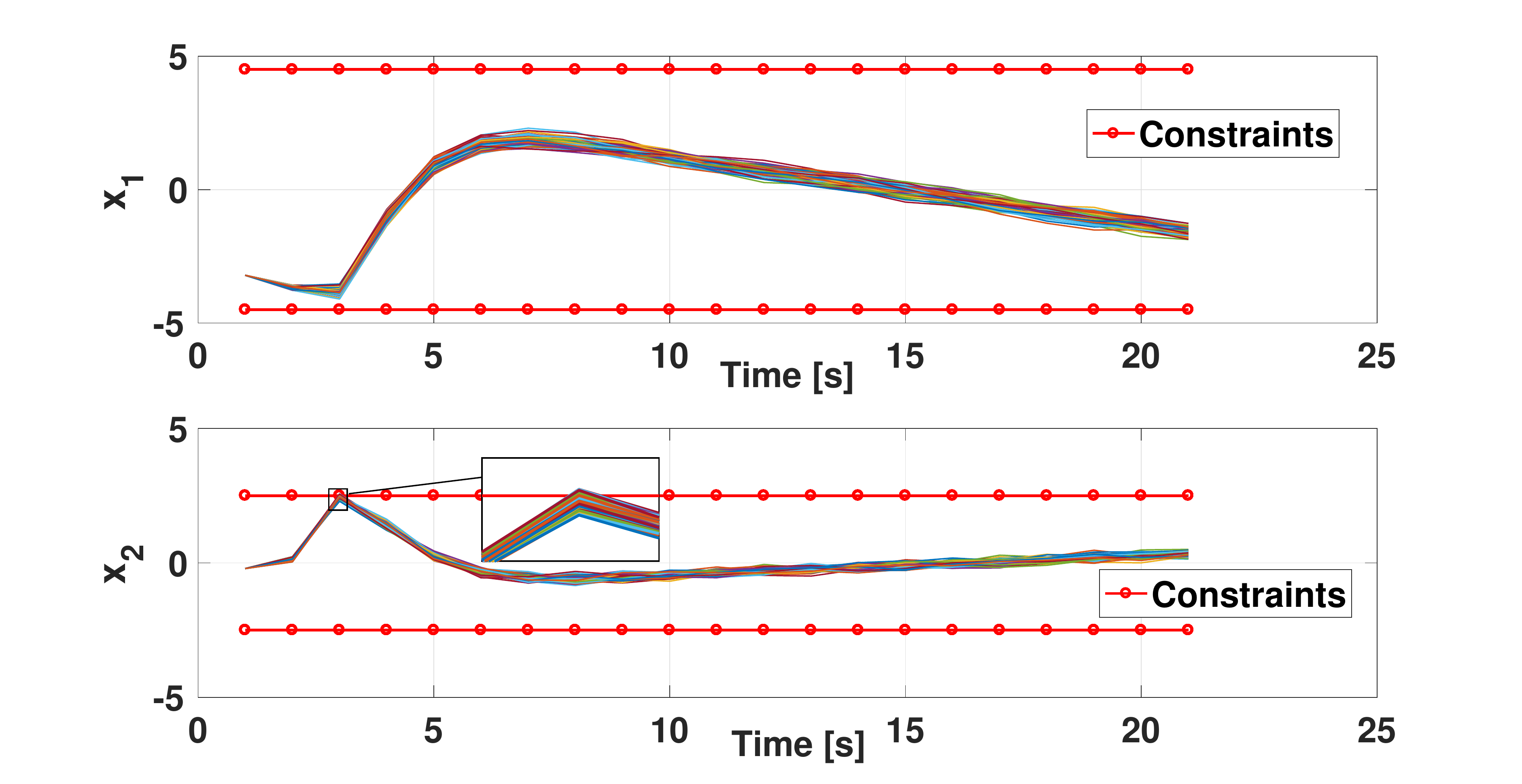}
    \caption{Monte Carlo Simulation of Stochastic Case}
    \label{Fig:MC_states_DIS}
\end{figure}
The closed-loop costs, $\sum \limits_{t=0}^{\infty}  \ell({x}^\star_{t|t}, v^\star_{t|t})$, of both the algorithms for the above Monte Carlo Simulations (under identical disturbance ($w_t$) realizations, $\theta^a_t$, initial conditions and MPC horizon lengths for both algorithms) are compared in Fig.~\ref{Fig:cost_compare}. 
\begin{figure}[t]
    \centering
	\includegraphics[width= 12cm]{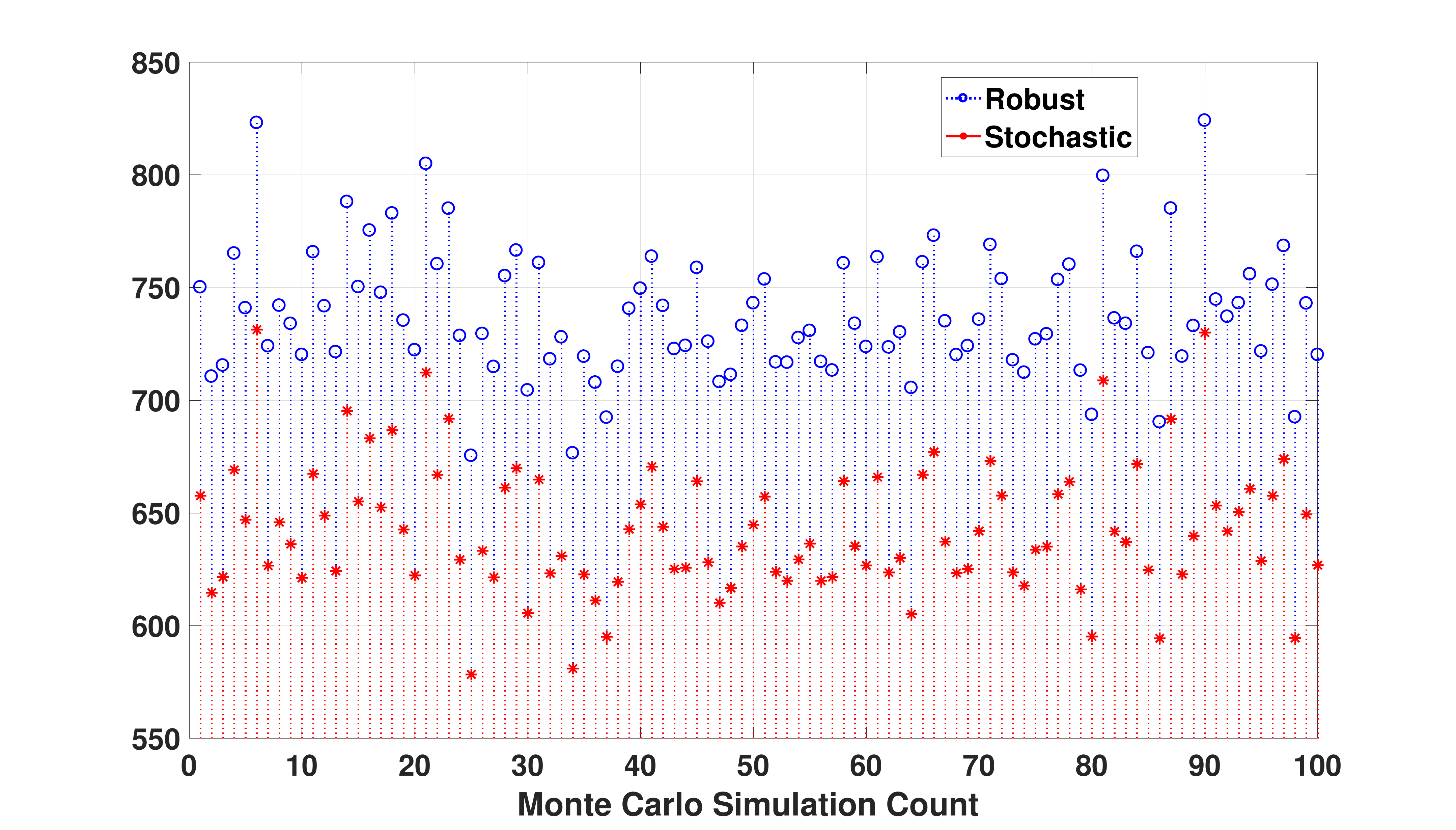}
    \caption{Comparison of Closed-Loop Cost $\textstyle\sum \limits_{t=0}^{\infty}  \ell({x}^\star_{t|t}, v^\star_{t|t})$}
    \label{Fig:cost_compare}
\end{figure}
The Adaptive Stochastic MPC algorithm delivers a reduction of $13 \%$ in average closed-loop cost compared to Adaptive Robust MPC. This indicates performance gain at the expense of hard constraint violations. 

Additionally, Fig.~\ref{Fig:cost_evol} shows the closed-loop MPC cost at each time step $t$, $\ell(x^\star_{t|t}, v^\star_{t|t})$, plotted over $100$ different trajectories, for the entire length of simulation duration.
\begin{figure}[t]
    \centering
	\includegraphics[width= 12cm]{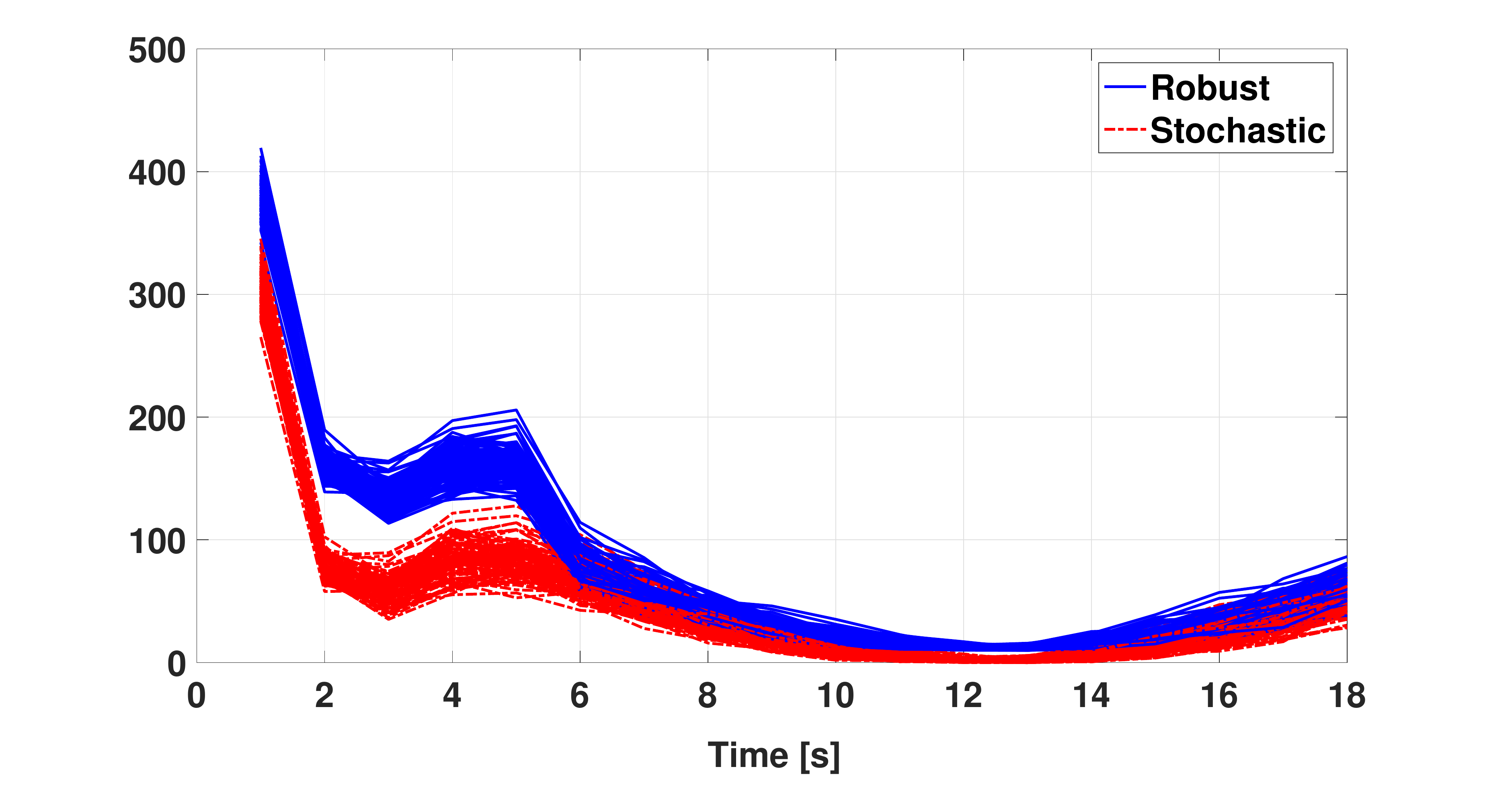}
    \caption{Stage Costs ($\ell(x^\star_{t|t}, v^\star_{t|t})$) in Closed-Loop}
    \label{Fig:cost_evol}
\end{figure}
It can be inferred from Fig.~\ref{Fig:cost_evol} that the bulk of the closed-loop cost reduction due to Adaptive Stochastic MPC over all simulated trajectories seen in Fig.~\ref{Fig:cost_compare}, occurs near the constraint violation zone (between $2$ seconds to $4$ seconds). 
\section*{Acknowledgement}
We acknowledge Ugo Rosolia for helpful discussions and reviews. This work was partially funded by Office of naval Research Grant ONR-N00014-18-1-2833.
 
\section{Conclusions}\label{sec:concl}
In this paper we proposed an adaptive MPC framework that handles both robust and probabilistic constraints. A Set Membership Method based approach is used to learn a bounded and time varying offset uncertainty in the model with available data from the system. We proved recursive feasibility and input to state stability of resulting MPC algorithms in presence of bounded additive disturbance/noise. We showed the validity and efficacy of the proposed approaches in detailed numerical simulations. 

\section*{Appendix} 

\subsection*{Proof of Proposition~1}
We prove Proposition~1 using induction, following the proof of the same in \cite{2017arXiv171207548T}. At time step $t=0$ we know that $\Theta_0 = \Omega$ and from Assumption~\ref{ass:omega_b}, $\Omega$ is nonempty and $\theta^a_0 \in \Omega$. Now using inductive argument, let us assume that the claim holds true for some $t\geq 0$. That is, for some nonempty $\Theta_t$, we have $\theta^a_t \in \Theta_t$. Now we must prove $\Theta_{t+1} \neq \phi$ and $\theta^a_{t+1} \in \Theta_{t+1}$. Let us define the following matrices:
\begin{align*}
    &\mathcal{H}^\theta_t = [(\mathcal{H}^\theta_0)^\top, (\bar{\mathcal{H}}^\theta_t)^\top]^\top \in \mathbb{R}^{r_t \times p},\\ &{h}^\theta_t = [({h}^\theta_0)^\top, (\bar{h}^\theta_t)^\top]^\top \in \mathbb{R}^{r_t},\\
    &{\Delta h}^\theta_t = [(\mathbf{0}_{r_0})^\top, (\Delta\bar{h}^\theta_t)^\top]^\top \in \mathbb{R}^{r_t},
\end{align*}
where $r_t = r_0 + 2t,~\forall t\geq 0$ is the number of faces of the Feasible Parameter Set polytope $\Theta_t$. Now from Assumption~\ref{ass:omega_b} we know: 
\begin{align}\label{eq:con_proof1}
    \mathcal{H}^\theta_0 \theta^a_{t+1} \leq h^\theta_0,
\end{align}
and from inductive assumptions we know that $\bar{\mathcal{H}}^\theta_t \theta^a_t \leq \bar{h}^\theta_t$. Therefore, we can ensure the following holds:
\begin{align}\label{eq:con_proof2}
    \bar{\mathcal{H}}^\theta_t \theta^a_{t+1} \leq \bar{h}^\theta_t + \Delta \bar{h}^\theta_t.
\end{align}
Moreover, we know that:
\begin{subequations}\label{eq:con_proof3}
\begin{align}
    &-E \theta^a_{t+1} \leq -x_{t+1}+Ax_t+Bu_t+\bar{w}-\ubar{\nu}, \\ 
    &+E \theta^a_{t+1} \leq x_{t+1}-Ax_t-Bu_t+\bar{w}+\bar{\nu}.
\end{align}
\end{subequations}
Hence, from \eqref{eq:con_proof1}, \eqref{eq:con_proof2} and \eqref{eq:con_proof3} we can have, $\mathcal{H}^\theta_{t+1} \theta^a_{t+1} \leq h^\theta_{t+1}$, where
\begin{align*}
\mathcal{H}^{\theta}_{t+1} & = [ (\mathcal{H}^{\theta}_0)^\top, (\bar{\mathcal{H}}^\theta_t)^\top, -E^\top,+E^\top]^\top \in \mathbb{R}^{r_{t+1} \times p},\\ 
    h^{\theta}_{t+1} & = \begin{bmatrix} h^{\theta}_0\\  \bar{h}^{\theta}_t + \Delta \bar{h}^{\theta}_t \\ -x_{t+1}+Ax_t+Bu_t+\bar{w}-\ubar{\nu} \\ x_{t+1}-Ax_t-Bu_t+\bar{w}+\bar{\nu} \end{bmatrix} \in \mathbb{R}^{r_{t+1}}.
    \end{align*}
This proves that $\Theta_{t+1}$ is nonempty and contains the actual offset uncertainty $\theta^a_{t+1}$ at the $(t+1)$-th time step. This concludes the proof.
\subsection*{Proof of Proposition~\ref{prop:offset_error}}
Utilizing the contraction property of Euclidean projection in \eqref{eq:proj_theta} similar to \cite{lorenzenAutomaticaAMPC}, we can write
\begin{align*}
 \frac{1}{\mu} \Vert \bar{\theta}_{t+1}  & - \theta^a_{t+1} \Vert ^2 - \frac{1}{\mu} \Vert \bar{\theta}_t - \theta^a_t \Vert ^2  \leq \frac{1}{\mu} \Vert \tilde{\theta}_{t+1} - \theta^a_{t+1} \Vert ^2 - \frac{1}{\mu} \Vert \bar{\theta}_t - \theta^a_t \Vert ^2,
\end{align*}
where $\Vert \cdot \Vert$ is the Euclidean norm. This gives 
\begin{align*}
& \frac{1}{\mu} \Vert \bar{\theta}_{t+1}  - \theta^a_{t+1} \Vert ^2 - \frac{1}{\mu} \Vert \bar{\theta}_t - \theta^a_t \Vert ^2 \\
& \leq \frac{1}{\mu} \Vert \tilde{\theta}_{t+1} - \theta^a_{t} \Vert ^2 + \frac{2}{\mu} (\tilde{\theta}_{t+1} - \bar{\theta}_t)^\top (\bar{\theta}_t - \theta^a_t) + \frac{2}{\mu} \tilde{\theta}^\top_{t+1} (\theta^a_t - \theta^a_{t+1}),\\
& = \frac{1}{\mu} \Vert \mu E^\top (\tilde{x}_{t+1|t} + w_t) \Vert ^2 + 2(\tilde{x}_{t+1|t} + w_t)^\top E (\bar{\theta}_t - \theta^a_t) + \frac{2}{\mu} \tilde{\theta}^\top_{t+1} (\theta^a_t - \theta^a_{t+1}),\\
& \leq \frac{1}{\mu} \Vert \mu E^\top (\tilde{x}_{t+1|t} + w_t) \Vert ^2 + 2(\tilde{x}_{t+1|t} + w_t)^\top E (\bar{\theta}_t - \theta^a_t) + \frac{2}{\mu} \Vert \tilde{\theta}_{t+1} \Vert \Vert (\theta^a_t - \theta^a_{t+1}) \Vert.
\end{align*}
Consider $\Omega$ and $\mathcal{P}$ sets from Assumption~1 and Assumption~2. Define $\sup_{\omega \in \Omega} \Vert \omega \Vert = \omega_M$ and $\sup_{p \in \mathcal{P}} \Vert p \Vert = p_M$. Then the above inequality can be written as \begin{align*}
& \frac{1}{\mu} \Vert \bar{\theta}_{t+1}  - \theta^a_{t+1} \Vert ^2 - \frac{1}{\mu} \Vert \bar{\theta}_t - \theta^a_t \Vert ^2 \\
& \leq \frac{1}{\mu} \Vert \mu E^\top (\tilde{x}_{t+1|t} + w_t) \Vert ^2 + 2(\tilde{x}_{t+1|t} + w_t)^\top E (\bar{\theta}_t - \theta^a_t) + \frac{2}{\mu} \omega_M p_M,\\
& \leq (\mu \Vert E \Vert^2 - 1) \Vert \tilde{x}_{t+1|t} + w_t \Vert ^2 - \Vert \tilde{x}_{t+1|t} \Vert ^2 + \Vert w_t \Vert ^2 + \frac{2}{\mu} \omega_M p_M,\\
& \leq - \Vert \tilde{x}_{t+1|t} \Vert ^2 + \Vert w_t \Vert ^2,
\end{align*}
since from Remark~2 we know $\frac{1}{\mu} > \Vert E \Vert ^2$, and we have used $x_{t+1} - \bar{x}_{t+1|t} = \tilde{x}_{t+1|t} + w_t$ and $\tilde{x}_{t+1|t} = E(\theta^a_t - \bar{\theta}_t)$. Summing both sides of the inequality from 0 to $\tilde{m}$ leads to a telescopic sum on the LHS, and we obtain, 
\begin{align*}
    & \frac{1}{\mu} \Vert \bar{\theta}_{\tilde{m}+1} - \theta^a_{\tilde{m}+1} \Vert ^2 + \sum \limits_{t=0}^{\tilde{m}} \Vert \tilde{x}_{t+1|t} \Vert^2  \leq \sum \limits_{t=0}^{\tilde{m}} \Vert w_t \Vert^2 + \frac{1}{\mu} \Vert \bar{\theta}_0 - \theta^a_0 \Vert ^2,  
\end{align*}
which, upon division by RHS on both sides gives
\begin{align*}
    \sup_{\tilde{m} \in \mathbb{N}, w_t \in \mathbb{W}, \bar{\theta}_0 \in \Omega} \frac{\sum \limits_{t=0}^{\tilde{m}} \Vert \tilde{x}_{t+1|t} \Vert^2 }{\frac{1}{\mu} \Vert \bar{\theta}_0 - \theta_0^a \Vert^2 + \sum \limits_{t=0}^{m} \Vert w_t \Vert^2 } \leq 1.
\end{align*}
\subsection*{Proof of Proposition~3}
From the definition of $\Theta_{t+1|t}$ in \eqref{eq:ol_fps}, we see that,
\begin{align*}
    &\mathcal{H}^\theta_{t+1} = [(\mathcal{H}^\theta_{t+1|t})^\top,-E^\top,+E^\top]^\top,\\     &{h}^\theta_{t+1} = \begin{bmatrix}{h}^\theta_{t+1|t}\\ -x_{t+1}+Ax_t+Bu_t+\bar{w}-\ubar{\nu} \\ x_{t+1}-Ax_t-Bu_t+\bar{w}+\bar{\nu}\end{bmatrix}.
\end{align*}
So $\Theta_{t+1|t+1} \subseteq \Theta_{t+1|t}$. Now, the matrices of the Predicted Feasible Parameter Sets at next time step, $\mathcal{H}^\theta_{k|t+1}$ and ${h}^\theta_{k|t+1}$ for all $k \in \{t+2,\dots,t+N-1\}$ are formed from $\mathcal{H}^\theta_{t+1}$ and ${h}^\theta_{t+1}$ by construction. Therefore, for all $k \in \{t+2,\dots,t+N-1\}$,
\begin{align*}
 &\mathcal{H}^\theta_{k|t+1} =  [(\mathcal{H}^\theta_{k|t})^\top,-E^\top,+E^\top]^\top,\\    
 &{h}^\theta_{k|t+1} = \begin{bmatrix}{h}^\theta_{k|t}\\ -x_{t+1}+Ax_t+Bu_t+\bar{w}-\ubar{\nu} \\ x_{t+1}-Ax_t-Bu_t+\bar{w}+\bar{\nu}\end{bmatrix},
\end{align*}
where $\mathcal{H}^\theta_{k|t}$ and ${h}^\theta_{k|t}$ are given by \eqref{eq:ol_fps}. So, for all $k \in \{t+2, \dots, t+N-1\}$, each of the sets for $\Theta_{k|t+1}$ are formed by the same inequalities which form $\Theta_{k|t}$, appended by two extra rows from the new measurement. Therefore, $\Theta_{k|t+1} \subseteq \Theta_{k|t}$ for all $k \in \{t+2,\dots,t+N-1\}$. Moreover, from \eqref{eq:fps_termc}, $\Theta_{t+N|t} = \Omega$. Using this,
\begin{align*}
    &\mathcal{H}^\theta_{t+N|t+1} =  [(\mathcal{H}^\theta_0)^\top, -E^\top, +E^\top]^\top,~\\ &{h}^\theta_{t+N|t+1} = \begin{bmatrix}{h}^\theta_0\\ -x_{t+1}+Ax_t+Bu_t+\bar{w}-\ubar{\nu} \\ x_{t+1}-Ax_t-Bu_t+\bar{w}+\bar{\nu}\end{bmatrix},
\end{align*}
and therefore $\Theta_{t+N|t+1} \subseteq \Theta_{t+N|t} = \Omega$ from the definition of $\Omega$ in \eqref{eq:Omega}. This proves the proposition. 
\subsection*{Dualization of Robust Problem}
In this section we show how the robust MPC problem \eqref{eq:MPC_R_fin} can be reformulated for efficient solving. The constraints in \eqref{eq:MPC_R_fin} can be compactly written with similar notations as \cite{Goulart2006}:
\begin{equation}\label{eq:com_con_appen_R}
    F_R\mathbf{v}_t +    \max_{\mathbf{w}_t,\pmb{\theta}_t} (F_R \mathbf{M}_t +  G_R)(\mathbf{w}_t + \mathbf{E}\pmb{\theta}_t) \leq c_R + H_R x_{t},
\end{equation}
where we denote, $\mathbf{v}_t = [v_{t|t}^\top, v_{t+1|t}^\top, \dots, v_{t+N-1|t}^\top]^\top \in \mathbb{R}^{mN}$, $\pmb{\theta}_t = [\theta_{t|t}^\top, \dots, \theta_{t+N-1|t}^\top]^\top\in \mathbb{R}^{pN}$ for all $\theta_{k|t} \in \Theta_{k|t}$, for all $k\in \{t,\dots,t+N-1\}$,  $\mathbf{E} = \textnormal{diag}(E,\dots,E) \in \mathbb{R}^{nN \times pN}$ and $\mathbf{w}_t = [w_{t|t}^\top, \dots, w_{t+N-1|t}^\top]^\top \in \mathbb{R}^{nN}$. The matrices above in \eqref{eq:com_con_appen_R} $F_R\in \mathbb{R}^{(sN+r_R)\times mN}, G_R \in \mathbb{R}^{(sN+r_R)\times nN}, c_R \in \mathbb{R}^{sN+r_R}$ and $H_R \in \mathbb{R}^{(sN+r_R)\times n}$ are obtained as:
\begin{align*}
&F_R = \begin{bmatrix}D&\mathbf{0}_{s \times m}&\cdots&\mathbf{0}_{s \times m}\\{C}B & D & \cdots & \mathbf{0}_{s \times m} \\ \vdots & \ddots & \ddots & \vdots \\ {C}A^{N-2}B & {C}A^{N-3}B & \cdots & D\\Y_RA^{N-1}B & Y_RA^{N-2}B & \cdots & Y_RB \end{bmatrix},\\
&G_R=\begin{bmatrix} \mathbf{0}_{s \times n}&\mathbf{0}_{s \times n}&\cdots & \mathbf{0}_{s \times n}\\{C}&\mathbf{0}_{s \times n}&\cdots&\mathbf{0}_{s \times n}\\ \vdots & \ddots & \ddots & \vdots \\ {C}A^{N-2}&{C}A^{N-3}&\cdots&\mathbf{0}_{s \times n}\\Y_RA^{N-1}&Y_RA^{N-2}&\cdots&Y_R \end{bmatrix},\\
&c_R = [b^\top,\dots,b^\top,z_R^\top],\\
&H_R=-[{C}^\top,({C}A)^\top,\dots,({C}A^{N-1})^\top,(Y_RA^N)^\top]^\top.
\end{align*}
For $k = \{t,\dots,t+N-1\}$, denote the set of polytopes $\mathbb{S}^R_{k|t}=\{w \in \mathbb{W},~\theta \in \Theta_{k|t}: S^R_{k|t}(w+ E\theta) \leq h^R_{k|t}\}$. Then we can define a polytope $\mathbb{S}_R = \{\mathbf{w}_t + \mathbf{E}\pmb{\theta}_t \in \mathbb{R}^{nN}: S^R (\mathbf{w}_t + \mathbf{E}\pmb{\theta}_t) \leq h^R, S^R \in \mathbb{R}^{a_R \times nN},~h^R \in \mathbb{R}^{a_R} \}$ with, $S^R = \textnormal{diag}(S^R_{t|t},\dots,S^R_{t+N-1|t}),~{h}^R = [(h^R_{t|t})^\top,\dots,(h^R_{t+N-1|t})^\top]^\top$.
Now \eqref{eq:com_con_appen_R} can be written with auxiliary decision variables $Z_R \in \mathbb{R}^{a_R \times (sN+r_R)}$ using \emph{duality} of linear programs as,
\begin{subequations}\label{eq:com_con_dual_R}
\begin{align*}
& F_R\mathbf{v}_t + Z_R^\top h_R \leq c_R + H_R x_t,\\
& (F_R \mathbf{M}_t + G_R) = Z_R^\top S_R,\\
& Z_R  \geq 0,
\end{align*}
\end{subequations}
which is a tractable linear programming problem that can be efficiently solved with any existing solver for real-time implementation of Algorithm~\ref{alg1}. 

\subsection*{Dualization of Stochastic Problem}
In this section we show how the stochastic MPC problem \eqref{eq:idealMPC_sto_trac} can be reformulated for efficient solving. The state constraints in \eqref{eq:idealMPC_sto_trac} can be compactly written as:
\begin{equation}\label{eq:com_con_max_appen_S}
    F^{(1)}_S\mathbf{v}_t +    \max_{\mathbf{w}_t,\pmb{\theta}_t} (F^{(2)}_S \bar{\mathbf{M}}_t +  G_S)\begin{bmatrix}\mathbf{w}_t + \mathbf{E}\pmb{\theta}_t\\ \pmb{\theta}_t \end{bmatrix}  \leq c_S + H_S x_{t}.
\end{equation}
where $\mathbf{v}_t$, $\pmb{\theta}_t$, $\mathbf{E}$, $\mathbf{w}_t$ are defined in Section~\ref{sec:control_param}, and $\bar{\mathbf{M}}_t = \begin{bmatrix}
\mathbf{M}_{t_{mN \times nN}}&\mathbf{0}_{mN \times pN}\\ \mathbf{0}_{pN \times nN}& \mathbf{0}_{pN \times pN} \end{bmatrix}$. The matrices $F^{(1)}_S\in \mathbb{R}^{(sN+r_S)\times mN}, F^{(2)}_S\in \mathbb{R}^{(sN+r_S)\times (mN+pN)}, G_S \in \mathbb{R}^{(sN+r_S)\times (nN+pN)}, c_S \in \mathbb{R}^{sN+r_S}$ and $H_S \in \mathbb{R}^{(sN+r_S)\times n}$ above in \eqref{eq:com_con_max_appen_S}, are obtained as:
\begin{align*}
&F^{(1)}_S =  \begin{bmatrix}(G\mathcal{B}_{t+1})^\top, (G\mathcal{B}_{t+2})^\top, \dots, (G\mathcal{B}_{t+N})^\top, (Y_S\mathcal{B}_{t+N})^\top\end{bmatrix}^\top,\\
&F^{(2)}_S = \textnormal{diag}(G,G,\dots,Y_S) \begin{bmatrix}\mathcal{B}_{t+1}&\mathbf{0}_{n \times p}&\cdots&\mathbf{0}_{n \times p}\\\mathcal{B}_{t+2} & \mathbf{0}_{n \times p} & \cdots & \mathbf{0}_{n \times p} \\ \vdots & \ddots & \ddots & \vdots \\ \mathcal{B}_{t+N} & \mathbf{0}_{n \times p}& \cdots & \mathbf{0}_{n \times p}\\\mathcal{B}_{t+N} & \mathbf{0}_{n \times p} & \cdots & \mathbf{0}_{n \times p} \end{bmatrix},\\
&G_S=\textnormal{diag}(G,G,\dots,Y_S)\begin{bmatrix} \begin{matrix}A \mathcal{C}_t\\ A \mathcal{C}_{t+1}\\ \vdots\\ A \mathcal{C}_{t+N-1}
\end{matrix}
  & \vline & \mathbf{E}_{nN \times pN} \\
  \hline
\begin{matrix}\mathcal{C}_{t+N} \end{matrix}
& \vline & \mathbf{0
}_{n \times pN}\end{bmatrix},\\
&c_S = [h_j-F^{-1}_{g_j^\top w}(1-\alpha_j),\dots,h_j-F^{-1}_{g_j^\top w}(1-\alpha_j),z_S^\top],\\
& ~~~~~~~~~~~~~~~~~~~~~~~~~~~~~~~~~~~~~~~~~~~~~~~\forall j \in \{1,\dots,s\},\\
&H_S=-[(GA)^\top,({G}A^2)^\top,\dots,({G}A^{N})^\top,(Y_SA^N)^\top]^\top.
\end{align*}
For $k = \{t,\dots,t+N-1\}$, denote the set of polytopes $\mathbb{S}^S_{k|t}=\{w \in \mathbb{W},~\theta \in \Theta_{k|t}: S^S_{k|t}\begin{bmatrix}w+ E\theta \\ \theta \end{bmatrix} \leq h^S_{k|t}\}$, then we can define a polytope $\mathbb{S}_S = \{\mathbf{w}_t,\pmb{\theta}_t \in \mathbb{R}^{(nN+pN)}: S^S \begin{bmatrix}\mathbf{w}_t+\mathbf{E}\pmb{\theta}_t \\ \pmb{\theta}_t \end{bmatrix} \leq h^S, S^S \in \mathbb{R}^{a_S\times (nN+pN)},~h^S \in \mathbb{R}^{a_S} \}$ with, $S^S = \textnormal{diag}(S^S_{t|t},\dots,S^S_{t+N-1|t}),~{h}^S = [(h^S_{t|t})^\top,\dots,(h^S_{t+N-1|t})^\top]^\top$. Now \eqref{eq:com_con_max_appen_S} can be written with auxiliary decision variables $Z_S \in \mathbb{R}^{a_S \times (sN+r_S)}$ using \emph{duality} of linear programs as:
\begin{subequations}\label{eq:com_con_dual_S}
\begin{align}
& F^{(1)}_S\mathbf{v}_t + Z_S^\top h_S \leq c_S + H_S x_t,\\
& (F^{(2)}_S \bar{\mathbf{M}}_t + G_S) = Z_S^\top S_S,~Z_S  \geq 0.
\end{align}
\end{subequations}
Moreover, \eqref{eq:idealMPC_sto_trac} imposes input constraints given by $H_u u_{k|t} \leq h_u$, for all $k \in \{t,\dots,t+N-1\}$, and for all $t \geq 0$. This can be written as:
\begin{align}\label{eq:input_dual_S}
    \bar{{H}}_u \mathbf{v}_t + \max_{\mathbf{w}_t, \pmb{\theta}_t} \bar{{H}}_u (\mathbf{M}_t(\mathbf{w}_t + \mathbf{E}\pmb{\theta}_t)) \leq \bar{h}_u,
\end{align}
where $\bar{{H}}_u = \textnormal{diag}(H_u,\dots,H_u) \in \mathbb{R}^{oN \times mN},~\bar{h}_u = [h_u^\top,\dots,h_u^\top]^\top \in \mathbb{R}^{oN}$. Similar to \eqref{eq:com_con_dual_S}, \eqref{eq:input_dual_S} can be written with auxiliary decision variables $Z_u \in \mathbb{R}^{a_R \times oN}$ as:
\begin{subequations}\label{eq:input_dual_S_com}
\begin{align}
& \bar{{H}}_u\mathbf{v}_t + Z_u^\top h_R \leq \bar{h}_u,\\
& \bar{{H}}_u \mathbf{M}_t = Z_u^\top S_R,~Z_u  \geq 0,
\end{align}
\end{subequations}
using $S_R$ and $h_R$ from the previous section. Clearly, \eqref{eq:com_con_dual_S} and \eqref{eq:input_dual_S_com} constitute a tractable linear programming problem that can be efficiently solved with any existing solver for real-time implementation of the Algorithm~\ref{alg2}. 

\renewcommand{\baselinestretch}{0.91}

\balance

\section*{References}
{
\printbibliography[heading=none, resetnumbers=true]
}

\end{document}